\documentclass[10pt,conference]{IEEEtran}
\IEEEoverridecommandlockouts
\usepackage{microtype}
\usepackage{amsthm}
\usepackage{amsmath,amssymb}
\usepackage{complexity}
\usepackage{xspace}
\usepackage{bm}
\usepackage[margin=1in]{geometry}
\usepackage{soul}
\usepackage{tabularx}
\usepackage{graphicx}
\usepackage[font=small,skip=0pt]{caption}
\usepackage{enumitem}

\newcommand{\subparagraph}{}

\usepackage{framed}
\usepackage{algorithm}
\usepackage{array}
\usepackage{xcolor}
\usepackage[noend]{algpseudocode}
\usepackage{enumitem}
\usepackage{multirow}
\usepackage{url}
\usepackage{titlesec}
\usepackage{threeparttable}
\usepackage{url}
\usepackage[hidelinks]{hyperref}

\newboolean{short}
\newcommand{\shortOnly}[1]{\ifthenelse{\boolean{short}}{#1}{}}
\newcommand{\onlyShort}[1]{\ifthenelse{\boolean{short}}{#1}{}}
\newcommand{\longOnly}[1]{\ifthenelse{\boolean{short}}{}{#1}}
\newcommand{\onlyLong}[1]{\ifthenelse{\boolean{short}}{}{#1}}

\titlespacing*{\section}
{0pt}{.5ex plus 1ex minus 1ex}{.5ex plus .8ex}
\titlespacing*{\subsection}
{0pt}{.5ex plus 1ex minus 1ex}{.5ex plus .8ex}

\date{}

\newcommand{\para}[1]{\vspace{0.1em}\noindent\textbf{#1.}~}

\newtheorem{theorem}{Theorem}

\newtheorem{lemma}[theorem]{Lemma}
\newtheorem{corollary}[theorem]{Corollary}

\begin{document}
\title{Latency, Capacity, and Distributed Minimum Spanning Tree}

\makeatletter
\newcommand{\linebreakand}{%
  \end{@IEEEauthorhalign}
  \hfill\mbox{}\par
  \mbox{}\hfill\begin{@IEEEauthorhalign}
}
\makeatother

\author{\IEEEauthorblockN{John Augustine}
\IEEEauthorblockA{\textit{Indian Institute of Technology Madras} \\
augustine@iitm.ac.in}
\and
\IEEEauthorblockN{Seth Gilbert}
\IEEEauthorblockA{\textit{National University of Singapore} \\
seth.gilbert@comp.nus.edu.sg}
\and
\IEEEauthorblockN{Fabian Kuhn}
\IEEEauthorblockA{\textit{University of Freiburg} \\
kuhn@cs.uni-freiburg.de}
\linebreakand %
\IEEEauthorblockN{Peter Robinson}
\IEEEauthorblockA{\textit{City University of Hong Kong} \\
peter.robinson@cityu.edu.hk}
\and
\IEEEauthorblockN{Suman Sourav}
\IEEEauthorblockA{\textit{National University of Singapore} \\
sourav@comp.nus.edu.sg}
}

\maketitle

\begin{abstract}
We study the cost of distributed MST construction in the setting where each edge has a latency and a capacity, along with the weight. Edge latencies capture the delay on the links of the communication network, while capacity captures their throughput (in this case, the rate at which messages can be sent). 
Depending on how the edge latencies relate to the edge weights, we provide several tight bounds on the time and messages required to construct an MST. 

When edge weights exactly correspond with the latencies,  we show that, perhaps interestingly, the bottleneck parameter in determining the running time of an algorithm is the total weight $W$ of the MST (rather than the total number of nodes $n$, as in the standard \textsf{CONGEST} model). That is, we show a tight bound of $\tilde{\Theta}(D + \sqrt{W/c})$ rounds, where $D$ refers to the latency diameter of the graph, $W$ refers to the total weight of the constructed MST and edges have capacity $c$. The proposed algorithm sends $\tilde{O}(m+W)$ messages, where $m$, the total number of edges in the network graph under consideration, is a known lower bound on message complexity for MST construction. We also show that $\Omega(W)$ is a lower bound for fast MST constructions.   

When the edge latencies and the corresponding edge weights are unrelated, and either can take arbitrary values, we show that (unlike the sub-linear time algorithms in the standard \textsf{CONGEST} model, on small diameter graphs), the best time complexity that can be achieved is $\tilde{\Theta}(D+n/c)$. However, if we restrict all edges to have equal latency $\ell$ and capacity $c$ while having possibly different weights (weights could deviate arbitrarily from $\ell$), we give an algorithm that constructs an MST in $\tilde{O}(D + \sqrt{n\ell/c})$ time.
In each case, we provide nearly matching upper and lower bounds.
\end{abstract}

\section{Introduction} \label{sec:intro}

\begin{table*}[htb]
\centering
\caption{MST Construction Results Summary}
\begin{tabular}{|c|c|c|}
\hline
Weights vs Latencies                     & Subcases                       & Round Complexity                     \\ \hline
Weights = latencies                      & Arbitrary weights/latencies    & $\tilde{\Theta}(D+\sqrt{W/c})$       \\ \hline
\multirow{2}{*}{Weights $\ne$ latencies} & Arbitrary weights \& latencies & $\Theta(D+n/c)$ 					 \\ \cline{2-3} 
                                         & Uniform latencies $\ell$       & $\tilde{\Theta}(D + \sqrt{n\ell/c})$   \\ \hline
\end{tabular}
\begin{tablenotes}
  \item 
   \center{$D$ is latency diameter; $n$ is number of nodes; $c$ is edge capacity and  $W$ is the total weight of the MST.}
\end{tablenotes}
\end{table*}

Construction of a minimum-weight spanning tree (MST) is one of the most fundamental problems in the area of distributed computing, and has been extensively studied (see~\cite{GHS, 4568149, Awerbuch:1987:ODA:28395.28421, Singh1995, GKP, KUTTEN199840, PR99, Elkin:2006:ULB:1139961.1146783, Elkin:2004:FDP:982792.982844,  King:2015:CIR:2767386.2767405, Khan2008, Pandurangan:2017:TMD:3055399.3055449, Elkin:2017:SDD:3087801.3087823, DBLP:conf/wdag/MashreghiK18} and references therein).

Much of this existing literature deals with the standard \textsf{CONGEST} model of communication in a synchronous setting \cite{Peleg:2000:DCL:355459} where all edges have a unit communication time (equaling round length) and in every round nodes can communicate with all their neighbors via $O(\log n)$ sized messages.\footnote{If all edges have uniform latency $\ell$, it is easy to see that these results also hold albeit with a scaling factor of $\ell$.}  
Alternatively, in asynchronous systems, edges can have an arbitrary amount of delay and there exists absolutely no synchronization among nodes. Neither scenario accurately depicts most real-world networks, where different connections have different speeds or latency (which may depend on distance, congestion, router speed, etc.) along with some level of synchrony. A lower latency implies faster packet delivery. With latencies, the communication is not totally asynchronous (as messages have fixed delays),  %
neither is it totally synchronous (as communication on different edges takes different amounts of time). To model a scenario with latencies in a totally synchronous world, one has to assume the worst case latency in the network as the round length. 
This assumption is quite wasteful, and oftentimes, the worst case latencies might be too high.

Latency, however, is not the only parameter that matters; the throughput (or the rate in which data can be pushed) of a communication link can often be the bottleneck in communication. Here, we describe the throughput as the capacity of a link, which we indicate as a fraction in $[0, 1]$.  If a link has capacity 1, then a new packet can be sent over the link in every time step (even if the earlier ones have not arrived yet).  If a link has capacity 1/10, then a new packet can only be sent once in every 10 rounds.  %

Both latency and capacity play an influential role in determining a network's performance (see \cite{Bakr:2002:ERT:571825.571864, 1607968} etc.). As an example, in \cite{Bakr:2002:ERT:571825.571864}, Bakr and Keidar show that the communication step metric (number of rounds required) fails to capture the behavior of an actual algorithm over the internet because of the existence of arbitrary latencies. Thereby, one of our goals is to correctly capture and model the performance of actual algorithms on networks.

When considering both latency and capacity, it can be tricky to determine the best way for two nodes to communicate.  If the nodes only have one packet to exchange, you want to find the path with minimum latency (though such a path may have more hops).  If they have a stream of packets to send, then you may want to find a path with high capacity.  Again, if you want to minimize the message complexity, then you might want a path that minimizes the number of hops.  Simultaneously optimizing these different parameters is a challenge!

In this paper, we focus on the problem of constructing an MST on graphs having edge latencies and capacity, giving algorithms and lower bounds for a variety of different cases. 
%

\iffalse

\para{Latency and Capacity}
The latency of an edge determines the time required to communicate via an edge.
Low latency on links imply faster message transmission whereas higher latency implies longer delays.
%
%
In fact, even if nodes are connected directly it might not be the fastest route for communication due to large latency of the link %
; often choosing a multi-hop lower latency path leads to faster algorithms. Basically, a hop optimal solution, no longer remains time optimal. Additionally, even communicating with a neighbor might not be feasible (if we care about time).
The capacity determines the rate at which messages can be sent over an edge.
We consider the capacity to be maximum, when a new message can be sent over a link in each round (non-blocking model \cite{8416344}). Similarly, capacity is minimized when there can exist only one message in the link at a particular time, i.e., a new message can be sent only when the previous message is delivered (blocking model \cite{8416344}). %

\fi
%

\para{Weighted \textsf{CONGEST} Model}
The network is modeled as a connected, undirected graph $G = (V,E)$ with $n = |V|$ nodes and $m=|E|$ edges. Each edge represents a bi-directional synchronous communication channel that has three attributes associated with it: latency, capacity, and weight. The weight provides the parameter over which we build an MST.  %

If an edge $(u,v)$ has capacity $c \in (0,1]$, it implies that $u$ can send a new message to $v$ in every $1/c$ rounds, i.e., if  $u$ had sent a message to $v$ in round $r$, then it can send the next message to $v$ only in round $r+(1/c)$.
For simplicity, we assume that the rate at which data can be sent remains constant throughout the network, i.e., all edges have the same capacity $c$ and $1/c$ is an integer. 

If an edge $(u,v)$ has latency $\ell$, it implies that it requires $\ell$ rounds for a message to be sent from $u$ to $v$ (or vice versa). 
We assume that each edge's latency and weight are integers. (If not, they can be scaled and rounded to the nearest integer.) 
Let $\ell_{min}$ be the minimum latency of any edge of the given graph $G$.  We assume that  $\ell_{min} \geq 1/c$. This ensures that, if there are no messages in transit over an edge, a node should be allowed to send a message over that edge.\footnote{It is reasonable to assume that an edge is not blocked for longer than its latency, so we require the capacity of an edge of latency $\ell$ to be at least $1/\ell$.}

Nodes know the value of $n$ and have unique ids.  Nodes can send $O(\log n)$-bit messages to all their neighbors in a particular round.
Nodes also know the latency, capacity, and weight of their adjacent edges; however, they are not aware of the ids of their neighbors (as in the KT0 model of computation).
The latency diameter $D$ of the graph $G$ refers to the graph diameter with latencies. Unless mentioned otherwise, $D$ denotes the latency diameter.

\para{Distributed MST Construction}
Given a connected, edge-weighted undirected graph $G=(V,E)$ with latencies and capacity $c$, the goal is to determine a set of edges $E$ that form a spanning tree of minimum weight. At the end of the distributed MST construction protocol, each node must know which of  its adjacent edges are in the computed MST.

\para{Results}
In this paper, we introduce the weighted \textsf{CONGEST} model with edge latencies and capacities that closely mimic real-world communication. We study the effects of latency and capacity in determining the time required for constructing an MST. Depending on how the edge latencies relate to the weights, we provide several tight bounds on both the time and messages required to construct an MST. 

We start by considering the case where edge weights also represent the latency on the edge.\footnote{For problems like MST or TSP (traveling salesman problem), the weights of the edges tend to model latency in many cases.} %
Surprisingly, for this case, the key parameter determining the delay due to congestion is the total weight $W$ of the constructed MST (rather than the total number of nodes $n$ in the graph as seen for the case of the standard \textsf{CONGEST} model on graphs with unit latencies). We propose an algorithm that constructs an MST in $\tilde{O}(D + \sqrt{W/c})$ rounds\footnote{The $\tilde{O}$, $\tilde{\Theta}$, and $\tilde{\Omega}$  notations hide polylogarithmic factors.} and with  $\tilde{O}(m + W)$ messages. %
Correspondingly, to make our bounds tight we show a lower bound of $\tilde{\Omega}(D+\sqrt{W/c})$ rounds. As part of the lower bound proof, we provide a simulation that relates the running time of an algorithm in this weighted \textsf{CONGEST} model with that in the standard \textsf{CONGEST} model. (c.f. Lemma \ref{lem:simulation}). In regard to the message complexity, we first show a lower bound of $\Omega(m)$ by leveraging on the results in \cite{Kutten:2015:CUL:2742144.2699440}. %
We also show that any algorithm that runs in a constant factor of the optimal running time, for a particular choice of constant, would require $\Omega(W)$ messages in the worst case.

Next, we consider the case where there is no correlation between latencies and weights.  In the standard \textsf{CONGEST} model, an MST can be constructed in $\tilde{O}(D' + \sqrt{n})$ time, where $D'$ refers to the graph diameter with unit latencies. However, in a network with arbitrary latencies, we show that sub-linear time MST construction is impossible.  Specifically, we give a lower bound of $\Omega(D+n/c)$ rounds for constructing an MST.  Correspondingly, we also give an algorithm that constructs an MST in $O(D+n/c)$ rounds and with $O(n^2)$ messages.  

A fundamental special case is where all edges have equal latency $\ell$.  We give a simultaneous time and message optimal algorithm (derived from~\cite{Elkin:2017:SDD:3087801.3087823}) that constructs an MST in $\tilde{O}(D + \sqrt{n\ell/c})$ time and with $\tilde{O}(m)$ messages.  This is faster than the expected $O(\ell)$ slowdown (achieved by scaling up the edge latencies from $1$ to $\ell$ in the standard \textsf{CONGEST} model), and this speed-up is achieved by exploiting the power of edge capacity through pipelining of messages. %

\para{Challenges}
There are some basic challenges that arise in designing MST algorithms for networks with latencies and capacities.  
There may be many edges that are just too expensive to use, and a node will never even know the identity or status of its neighbors on the other side of these edges.  Moreover, it may not be clear in advance which edges are too expensive to use, as that depends on various parameters, e.g., $D$ or $W$.  For example, when a node is trying to find a minimum weight outgoing edge of a component, it may never be able to find out whether a neighbor is in the same connected component.  Or as another example, our MST algorithms rely on collecting information on BFS/shortest path trees; yet in constructing the BFS tree, there are some edges that cannot be used.  How does a node know when the construction is complete?  In all our protocols, we must carefully coordinate the exploration of edges to avoid using expensive edges and to compensate for unknown information.

Most existing distributed MST algorithms that try to optimize both the time and messages complexities usually runs in two stages (see \cite{GHS,Pandurangan:2017:TMD:3055399.3055449, Elkin:2017:SDD:3087801.3087823}), and we also adopt a similar strategy). In the first stage, usually, the MST is built in a bottom-up fashion by merging MST-fragments (a connected subgraph of the MST). As and when these fragments become large, the cost of communicating on them increases, marking the beginning of the second stage, where algorithms use a BFS tree to further build the MST. This switching point is no longer as simple to determine, as it depends on various unknown parameters of the graph, e.g., $D$ and $W$.  Our algorithms have to on-the-fly determine the best point to switch between stages. 
Moreover, it is no longer the case that the same tree is good for both minimizing latency and message complexity. This makes the balancing problem even more difficult, if we want to maintain reasonable message complexity. 
A related problem shows up in the initial construction of the BFS/shortest path tree.  In a model with unit latency edges, there are a variety of strategies for electing a leader and using it to initiate a shortest-path tree (even with good message complexity~\cite{Kutten:2015:CUL:2742144.2699440}).  However, when links have arbitrary latencies, this becomes non-trivial, and we rely on a simple randomized strategy.

\para{Prior Work}  %
The problem of distributed computation of MST was first proposed in the seminal paper of Gallager, Humblet, and Spira \cite{GHS}, which presented a distributed algorithm for MST construction in $O(n\log n)$ rounds and with $O(m + n \log n)$ messages.
The time complexity was further improved in \cite{4568149} and subsequently to existentially optimal $O(n)$ by Awerbuch \cite{Awerbuch:1987:ODA:28395.28421}.
Existential optimality implies the existence of graphs for which $O(n)$ is the best possible time complexity achievable. 
These (and many subsequent works, including this) are based on a non-distributed variant of the algorithm of {B}or\r{u}vka \cite{NESETRIL20013}. 

In a pioneering work \cite{GKP} Garay, Kutten, and Peleg showed that the parameter that best describes the cost of constructing an MST is the graph diameter $D'$ (with unit latency edges), rather than the total number of nodes $n$. For graphs with sub-linear diameter, they gave the first sub-linear distributed MST construction algorithm requiring $O(D' + n^{0.614} \log^* n)$ rounds and $O(m + n^{1.614})$ messages. This was further improved to $O(D'+\sqrt{n}\log^* n)$ rounds and $O(m + n^{1.5})$ messages by  Kutten and Peleg \cite{KUTTEN199840}. Shortly thereafter, Peleg and Rubinovich \cite{PR99} showed that $\Omega(D' + \sqrt{n}/ \log n)$ time is required by any distributed MST construction algorithm, even on networks of small diameter ($D'=\Omega(\log n)$), establishing the asymptotic near-tight optimality of the algorithm of \cite{KUTTEN199840}. Consequently, the same lower bound of $\Omega(D' + \sqrt{n})$ was shown for randomized (Monte Carlo) and approximation algorithms as well \cite{DHK+12, Elkin:2006:ULB:1139961.1146783}.

The message complexity lower bound of $\Omega(m)$ was first established by Awerbuch \cite{Awerbuch:1987:ODA:28395.28421} for deterministic and comparison based randomized algorithms. In \cite{Kutten:2015:CUL:2742144.2699440}, Kutten et al. show that the lower bound holds for any algorithm, in the KT0 model of communication, where nodes do not know the ids of its neighbors. However, for general randomized algorithms, if the nodes are aware of the ids of their neighbors (KT1 model) the $\Omega(m)$ message complexity lower bound does not hold. In fact in \cite{King:2015:CIR:2767386.2767405}, King, Kutten, and Thorup give an MST construction algorithm with a message complexity of only $\tilde{O}(n)$, however this came at the expense of having time complexity of $\tilde{O}(n)$. For asynchronous networks, Mashreghi and King \cite{DBLP:conf/wdag/MashreghiK18} give an algorithm that computes an MST using only $o(m)$ messages.

More recently, for the KT0 model, Pandurangan et al. \cite{Pandurangan:2017:TMD:3055399.3055449} provide a randomized MST construction algorithm with time complexity $\Tilde{O}(D' + \sqrt{n})$ and message complexity $\Tilde{O}(m)$, which is simultaneously time an message optimal. Elkin \cite{Elkin:2017:SDD:3087801.3087823}, Haeupler et al. \cite{Haeupler:2018:RMD:3212734.3212737}, and Ghaffari and Kuhn \cite{ghaffari_et_al:LIPIcs:2018:9819} have since provided improved deterministic algorithms that achieve the same bounds (with improvements in logarithmic factors).

As discussed earlier, both latency and capacity play a significant role in determining a network's performance (see \cite{Bakr:2002:ERT:571825.571864, 1607968} etc.). In \cite{Awerbuch:1990:CAC:93385.93417}, Awerbuch et al. study the impact of transmission delays on several different distributed algorithms including MST construction (by using the methods of the pre-$\sqrt{n}$-era). %
There has also been some recent work by Sourav, Robinson, and Gilbert~\cite{8416344} on graphs with latencies. They looked at the problem of gossip, and developed a notion of weighted conductance that captured the connectivity of a graph with latencies.  They used this to analyze the cost of information dissemination in such graphs.

\section{Equal Weights and Latencies} \label{sec:correlated}

In this section, we consider the weights of the edges to be exactly equivalent to the edge latency (the results also hold if there exists a fixed relationship between the weights and the latencies). Unlike the case with unit latencies, where the running time of an algorithm depends on the total number of nodes $n$ (along with the diameter $D$), here we see that when weights do represent latencies the running time of any MST construction algorithm becomes dependent on the total weight $W$ of the MST (along with the diameter and the edge capacities). 

Specifically, we give an MST construction algorithm that runs in $\tilde{O}(D+\sqrt{W/c})$ time while sending $\tilde{O}(m+W)$ messages. %
We also provide the corresponding time complexity lower bound. 
For message complexity, we first show a lower bound of $\Omega(m)$. %
Thereafter, for fast MST algorithms, we show another lower bound of $\Omega(W)$. 

\subsection{Upper Bound} \label{subsec:ub-unequal}
In this section, we provide an algorithm for constructing an MST  when the edge latency of each edge matches with its edge weight requiring $\tilde{O}(D+\sqrt{W/c})$ time and $\tilde{O}(m+W)$ messages.

\para{Preliminaries} \label{sec:prelims}
We first introduce some notation. 
Given a graph $G$, let $T$ be the (unique) MST of $G$. A fragment $F$ of $T$ is defined as a connected subgraph of $T$, that is, $F$ is a rooted subtree of $T$. The root of the fragment is called the fragment leader. 
Each fragment is identified by the id of the fragment leader, and each node knows its fragment's id (enforced as an invariant by the algorithm). An edge is called an outgoing edge of a fragment if one of its endpoints lies in the fragment and the other does not. The minimum-weight outgoing edge (MOE) of a fragment 
$F$ is the edge with minimum weight among all outgoing edges of $F$.
\noindent\textbf{Algorithm for Equal Weights and Latencies} \label{subsubsec:ubcorrelated} \\
In order to obtain an algorithm that not only gives the optimal time complexity but also a reasonable message complexity, %
we base our idea on the Elkin's algorithm \cite{Elkin:2017:SDD:3087801.3087823} for graphs with unit latencies. 
The premise of our algorithm is simple (and is in similar flavor with many of the existing distributed MST algorithms for graphs with unit latencies~\cite{GHS, GKP, Pandurangan:2017:TMD:3055399.3055449, Elkin:2017:SDD:3087801.3087823}) where the MST is built in a bottom up fashion by merging fragments. Over the iterations, fragments merge with one another until there remains a single MST fragment which is the required MST. In most cases, e.g. \cite{Pandurangan:2017:TMD:3055399.3055449} and \cite{Elkin:2017:SDD:3087801.3087823}, where there is a trade-off between the time and messages, there exists a balance between just building the fragment bottom-up (which can lead to large sized fragments, over which it would be too costly to communicate) and communicating via an external structure (usually a BFS tree of more manageable size). The main bottleneck of using the external structure often arises from possible congestion. %
We adopt a similar strategy (with a few key differences in implementation) of first building the fragments bottom-up and thereafter switching to use an external structure and broadly divide our algorithm into two stages; the local aggregation stage and the global aggregation stage.%

In the \textit{local aggregation} stage, by communicating via the fragment edges, we first build an initial set of fragments called ``base fragments" that satisfy a certain condition. This condition (as shown later) helps in determining the number of base fragments at the end of the first stage which in turn determines the maximum possible slowdown due to congestion in the second stage i.e. the \textit{global aggregation} stage. In this stage, to account for arbitrary latencies, subsequent fragment mergings are done by aggregating information collected from base fragment leaders over a shortest path tree. Optimum time complexity is obtained by balancing the cost of local aggregation and the global aggregation. The key idea in either stage is to have sufficiently frequent fragments mergings such that the algorithm terminates in a small number of iterations.

\para{Challenges and Countermeasures} With arbitrary latencies, if we use the previous approaches (of \cite{Pandurangan:2017:TMD:3055399.3055449} or \cite{Elkin:2017:SDD:3087801.3087823}) and focus on building base fragments up to a certain latency diameter, we can no longer say anything useful regarding the fragment size\footnote{This is because of the fact that with arbitrary latencies, a fragment with fragment diameter $k$ could have only $2$ nodes if it only contains an edge with latency $k$. Alternatively, it could contain up to $k$ nodes.} and hence the number of base fragments created. In the worst case, there can be up to $n/2$ base fragments.
Since an MOE (minimum-weight outgoing edge) can have an arbitrary latency (and unlike previous cases this choice influences the cost of creating the MST), we cannot allow arbitrary mergings and need to be careful in determining the specific mergings that are allowed.
Basically, if we do not distinguish between MOE edges (of different latencies), the cost of communicating within a fragment may become too high. Additionally, if we set too strict criteria for merging, fragments may not merge regularly enough, requiring a larger number of iterations. To achieve optimal time complexity and minimal message complexity, there has to be a balance between the cost of local aggregation (communicating within a fragment) with the cost of global aggregation (determined by possible congestion caused by the number of created base fragments). This balance, unlike the unit latency case, depends on the total weight $W$ of the MST. Nodes, without knowing the value of $W$ would have to determine on the fly the exact balance so as to decide when to switch from the local aggregation stage to the global aggregation stage.

To get around these issues, firstly, instead of controlling the growth of fragment diameter directly, we limit the total weight up to which fragments can grow in a particular iteration $i$.
Additionally, here in a particular iteration $i$ of the local aggregation algorithm, we only allow edges of weight (latency) $2^i$ or less to be used for fragment mergings. %
As such, for simplicity, one can view the local aggregation algorithm in iteration $i$ to be running on the sub-graph $G(2^i)$ of the given graph $G$, that only consists of the edges of latency $\leq 2^i$. Finally, we use a guess and double technique to determine the ideal balance between the number and the diameter of base fragments. Initially, we present the algorithm assuming that the nodes know the value of $W$ and later show that even if $W$ is not known, the MST can be computed through a guess and double strategy.

\onlyLong{
Notice that with arbitrary edge latencies, a hop-optimal solution no longer implies a cost-optimal solution. For example, the diameter of a BFS tree might be greater than the diameter of the graph, making a BFS tree unsuitable for algorithms requiring optimal time complexity. Therefore, we use a shortest path tree rather than a BFS tree. However, constructing a deterministic shortest path tree with arbitrary latencies is also non-trivial, especially if we want an algorithm with low message complexity. Additionally, due to the lack of synchrony, another challenge here while upcasting is to ensure that each node has all the required information to determine the correct edge to upcast.} %

\onlyLong{
\para{Shortest Path Tree Construction and Leader Election}  \label{spt}
To determine a shortest path tree rooted at some node, we use a simple randomized flooding mechanism:
Initially, each node becomes \emph{active} with probability $4\log(n) / n$ and if it is active, it forms the root of a shortest path tree by entering the \emph{exploration phase}.
Then, each active node $u$ broadcasts a \texttt{join} message carrying its ID to its neighbors who in turn propagate this message to their neighbors and so on.  The tree construction cannot wait to terminate until every edge is explored; instead, a counting mechanism is used to determine when the tree is spanning.
Therefore each root node $u$ sends out a \texttt{count} message (carrying its ID) in round $2^i+1$, for each $i=k,k+1,k+2,\dots$ until $u$ exits the \emph{exploration phase}, where $k$ is an integer such that $2^k \ge 1/c$.
The \texttt{count} messages propagate through $u$'s (current) tree until they reach the leaf nodes, who initiate a convergecast back to the root with a count of $1$. 
When a node receives the convergecast from its children it forwards the accumulated count to its parent in the shortest path tree.  

Since multiple nodes are likely to become active and start this process, eventually a node $w$ will receive \texttt{join} messages originating from distinct root nodes. 
In that case, $w$ joins the shortest path tree rooted at the node with the maximal ID.
If it has already joined some other shortest path tree previously that is rooted at a node with a smaller ID $v'$, it simply stops participating in that tree and responds to messages from that tree by sending a \texttt{disband} reply carrying $v'$'s ID.
A disband message propagates all the way to the root $v'$ who in turn becomes \emph{inactive} and exits the exploration phase. 

If an active node is still in the exploration phase when it receives a \texttt{count} message carrying a count of $n$, it stops exploring and broadcasts a \texttt{done} message through its tree.

\begin{lemma} \label{lem:spt}
In the weighted {CONGEST} model, when edges have arbitrary latencies, there exists an algorithm to elect a leader $u$ and construct a shortest path tree rooted at $u$ in $O(D)$ time using $O(m\log n )$ messages with high probability. 
\end{lemma}

\begin{proof} 
The time complexity depends on the time until every node has exited the exploration phase. 
Observe that the active node $u$ with the maximum ID will have integrated all nodes into its shortest path tree in $D$ time and, by the description of an algorithm, once a node joins $u$'s tree, it does not leave it.
Moreover, $u$ becomes aware that its tree has included all nodes within $O(D)$ additional rounds due to the \texttt{count} messages.

To see that the message complexity bound holds, observe that, with high probability, there are $\Theta(\log n)$ active nodes and each active node may initiate the construction of a shortest path tree.
\end{proof}
}

We call a fragment $F$ in iteration $i$ a \textit{blocked fragment} if all its adjacent MOE edges (including its own) has latency/weight greater than $2^i$, and therefore it cannot merge with any other fragment in iteration $i$ (c.f. Algorithm \ref{algo:CGHS2}). %
All the other fragments, that can still merge are called as \textit{non-blocked fragments}. %
As the growth of some fragments can now be blocked, another challenge is to regulate the number of base fragments. 
More base fragments leads to higher time complexity while accounting for congestion in communicating via the shortest path tree  
\onlyShort{(with arbitrary latencies a hop optimal solution need not necessarily imply a cost optimal solution, and therefore we use a shortest path tree rather than a BFS tree). However, constructing a deterministic shortest path tree with arbitrary latencies is also non-trivial, especially if we want an algorithm with low message complexity.} %

\para{Algorithm}
\onlyShort{The algorithm begins with the construction of a shortest path tree $\tau$, the procedure for which is described in the full version of the paper \cite{augustine2019latency}. Note that, from the shortest path tree it is easy to obtain a 2-approximation for the latency diameter $D$ of the graph, and based on $D$, we divide our algorithm in to two cases.  This is done in order to obtain a time optimal algorithm having minimal message complexity. First, we describe the algorithm for the case where $D \leq \sqrt{W/c}$ and thereafter consider the case where $D > \sqrt{W/c}$. As discussed earlier, for each case the algorithm is two-staged, consisting of a local aggregation stage and a global aggregation stage. }
\onlyLong{The algorithm begins with the construction of a shortest path tree $\tau$. Note that, from the shortest path tree $\tau$, it is easy to obtain a 2-approximation for the latency diameter $D$ of the graph, and based on $D$, we divide our algorithm in to two cases. This is done in order to obtain a time optimal algorithm having minimal message complexity. First, we describe the algorithm for the case where $D \leq \sqrt{W/c}$ and thereafter consider the case where $D > \sqrt{W/c}$. As discussed earlier, for each case the algorithm is two-staged, consisting of a local aggregation stage and a global aggregation stage. }

\enlargethispage{-\baselineskip}
\para{Local Aggregation Stage} (Creating Base Fragments)
Local aggregation begins with each node as a singleton fragment. Thereafter in every iteration, each fragment finds its MOE and some fragments are merged along their MOE in a controlled and balanced fashion, until the base fragments of the required total weight are obtained. Mergings are done by determining the MOE for each node of a fragment and convergecasting only the lightest edge seen up to the fragment leader (using only the fragment edges). The fragment leader decides the overall MOE for the fragment, and the merging (if occurs) occurs over this MOE. 
The guarantee here is twofold; first that the fragments merge sufficiently regularly (i.e the number of fragments reduces by at least half in each iteration) and secondly, the total number of blocked fragments at the end of local aggregation is not too much.

Consider $\mathcal{F}_i$ to be the set of fragments $\{F_1, F_2, \dots\}$ at the start of the $i^{th}$ iteration, e.g. $\mathcal{F}_1$ consists of $n$ singleton fragments. 
For the purpose of analysis, we define a fragment graph $\mathcal{H}_i = (\mathcal{F}_i , M_i)$ as follows. For a particular iteration $i$, its fragment graph $\mathcal{H}_i$ consists of the vertices
$\mathcal{F}_i = \{{F}_1 , \dots , {F}_k \}$, where each ${F}_j (1 \le j \leq k)$ is a fragment at the start of iteration $i$ of
the algorithm. The edge set $M_i$ of $\mathcal{H}_i$ is obtained by contracting the vertices of each
fragment ${F}_j \in \mathcal{F}_i$ to a single vertex in $\mathcal{H}_i$ and removing all resulting self-loops of
$\mathcal{H}_i$, leaving only the MOE edges in set $M_i$. Also, let $M_i'$ be the set of edges chosen by the algorithm over which fragment mergings happen in the iteration $i$.
Notice that the fragment graph $\mathcal{H}_i$ is, in fact, a tree which is not explicitly constructed in the algorithm, rather it is just a construct to explain it.
The pseudocode of the local aggregation algorithm for the case where $D \leq \sqrt{W/c}$ is shown in Algorithm \ref{algo:CGHS2}; it uses a similar techniques as the controlled-GHS algorithm in \cite{Gopalsurvey} (also see \cite{GKP}, \cite{Kutten:2015:CUL:2742144.2699440}, and MST forest construction in \cite{Elkin:2017:SDD:3087801.3087823}).

\onlyLong{
\begin{figure*}
	\centering
	\includegraphics[scale=0.7]{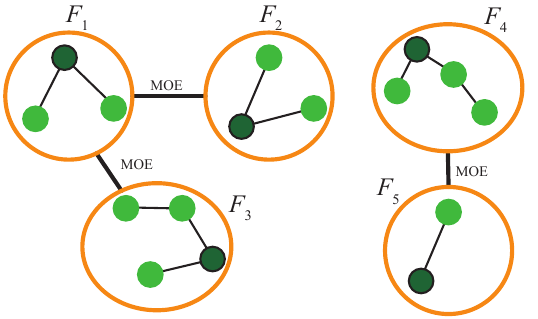}
	\caption[Parallel Fragment Merging]{An example of parallel fragment merging via MOE edges in local aggregation stage.}
	\vspace{-1em}
\label{fig:moe}
\end{figure*}
}

\begin{algorithm} 
\setlength{\columnsep}{10pt}
\begin{algorithmic}[1]
\State $\mathcal{F}_1 = (V, \phi)$ // initial set consisting of $n$ (singleton) fragments.
\For {$i = 0, \dots , \lceil \log \sqrt{W/c} \rceil$}
\State Each fragment $\mathcal{F} \in \mathcal{F}_i$ of weight at most $2^i$ determines the MOE of $\mathcal{F}$ and if the weight of the MOE edge is $\leq 2^i$, the fragment adds the edge to the candidate set $M_i$.
\State Find a maximal matching $M_i' \subseteq M_i$ in the fragment graph $\mathcal{H}_i = (\mathcal{F}_i , M_i )$.
\State If $\mathcal{F} \in \mathcal{F}_i$ of total weight at most $2^i$ has no incident edge in $M_i'$, it adds its
MOE edge into $M_i'$, iff its MOE edge has weight $\leq 2^i$.
\State  $\mathcal{F}_{i+1}$ is obtained by merging all the fragments along the edges selected in $M_i'$.
\EndFor

\end{algorithmic}
\caption{Local Aggregation: Outputs at most $2\sqrt{c{W}}$ base fragments of diameter at most $O(\sqrt{W/c})$.}\label{algo:CGHS2}
\end{algorithm}

\onlyLong{
\begin{lemma} \label{lem:frag_diam}
At the start of the $i^{th}$ iteration, each fragment has a diameter of at most $O(2^i)$. %
\end{lemma}

\begin{proof}

We show via an induction on the iteration number $i$, that, at the start of iteration $i$, the diameter of each fragment is at most $6 \cdot 2^i$.
The base case, i.e., at the start of iteration $0$, the statement is trivially true, since $6 \cdot 2^i = 6 \cdot 2^0 = 6$
which is greater than $0$, the total weight of a singleton fragment.
For the induction hypothesis, assume that the diameter of each fragment at the start of
iteration $i$ is at most $6 \cdot 2^i$. We show that when the $i^{th}$ iteration ends  (i.e., at the start
of iteration $i + 1$), the diameter of each fragment is at most $6 \cdot 2^{i+1}$.
Fragments grow by merging with other fragments over a matching MOE edge of the fragment graph. We know from the description of the algorithm (see Line $3$ of Algorithm \ref{algo:CGHS2}), 
that at least one of the fragments taking part in the merging has a diameter of at most $2^i$ (since only fragments with weight at most $2^i$ find MOE edges), however, that MOE edge might lead to a fragment with larger diameter (at most $6 \cdot 2^i$). 
Additionally, apart from the fragment joining via the matching edge, some other fragments with weight (also the diameter) at most $2^i$ can possibly join with either of these merging fragments of the matching edge, if they did not have any adjacent matching edge thereby increasing their diameter (see Line $5$ of Algorithm \ref{algo:CGHS2}). 

However, these unmatched fragments each have a weight/diameter of at most $2^i$, and therefore cannot increase the overall diameter by too much. The resulting diameter of the fragment at the end of
phase $i$ is at most $6 \cdot 2^i + 3 \cdot 2^i + 3 \cdot 2^i $,  
where at most one fragment (joined via the matched edge of $\mathcal{H}_i$) has a diameter of $6 \cdot 2^i$ and the other fragments contribute a weight/diameter of at most $2^i$ each. Finally, these are joined by $3$ MOE edges (of weight at most $2^i$, and therefore can possibly contribute at most  $3\cdot 2^i$ to the diameter of the merged fragment). Therefore, the diameter at the end of iteration $i$ is at most $6 \cdot 2^i + 3 \cdot 2^i + 3 \cdot 2^i \leq 6 \cdot 2^{i+1}$, for $i \geq 1$, completing the proof by induction.
\end{proof}

The corollary below follows as the local aggregation algorithm runs for $\lceil \log \sqrt{W/c}\rceil$ iterations.

\begin{corollary} \label{cor:frag_diam}
At the end of the local aggregation algorithm each fragment has diameter at most $O(\sqrt{W/c})$. 
\end{corollary}

\begin{lemma} \label{lem:frag_weight}
At the start of the $i^{th}$ iteration, each non-blocked fragment has a total weight of at least $2^i$.
\end{lemma}

\begin{proof}
We prove the above lemma via an induction on the iteration number $i$. For the base case, i.e. at the start of the iteration $0$, there exist only singleton fragments which are of weight at least $2^0$. For the induction hypothesis we assume that the statement is true for iteration $i$, i.e. at the start of iteration $i$, the total weight of each fragment is at least $2^{i}$ and show that the statement also holds for iteration $i+1$, i.e. at the start of iteration $i+1$, the total weight of each fragment is at least $2^{i+1}$. To show this, consider all the non-blocked fragments in iteration $i+1$, each fragment would either have weight $\geq 2^{i+1}$ or less than that. For fragments with weight $\geq 2^{i+1}$, the lemma is vacuously true. For the second case, where the fragment weight is $< 2^{i+1}$, we know from the algorithm (see line $3$ and $5$ of \ref{algo:CGHS2}), that all such fragments merge with at least one more fragment. This other fragment has a weight at least $2^i$ (from the induction hypothesis), and therefore the total weight of the resulting fragment at least doubles  i.e., becomes at least $2^{i+1}$, thus, proving the lemma.
\end{proof}

\begin{lemma} \label{lem:frag_number}
The number of fragments remaining at the start of the $i^{th}$ iteration is
at most $2(W/2^i)$. %
\end{lemma}

\begin{proof}
Each remaining fragment at the beginning of iteration $i$ is either a non-blocked fragment, or a blocked fragment.
We know from Lemma \ref{lem:frag_weight}, at the start of iteration $i$, each non-blocked fragment has a total weight at least $2^i$. Since fragments are disjoint and the total weight of all the fragments is $W$ (weight of the MST), this implies that the number of non-blocked fragments at the start of iteration $i$, is at most $W/2^i$.
Additionally, each blocked fragment would have all its adjacent MOE edges of weight $\geq 2^i$ (otherwise, it would not be blocked).  The maximum possible number of MOE edges of weight $\geq 2^i$ that can exist, is at most $W/2^i$, since each of the MOE edges would be a part of the MST and the total weight of the MST is $W$. This implies that the number of blocked fragments at the start of iteration $i$, is also at most $W/2^i$.
Therefore, the total number of fragments remaining at the start of iteration $i$, is the sum of the non-blocked and the blocked fragments, which is $2(W/2^i)$.
\end{proof}

\begin{corollary} \label{cor:frag_number}
At the end of the local aggregation algorithm, the number of fragments remaining is at most $2\sqrt{c{W}}$, where $c$ refers to the edge capacity and $W$ refers to the total weight of the MST.
\end{corollary}

}

We give the following lemma that determines the number of base fragments and the time required for creating them. \onlyShort{Due to the lack of space, the proof of the following lemma as well as some of the other missing proofs have been deferred to the full version of the paper \cite{augustine2019latency}.} %

\begin{lemma} \label{lem:CGHS_result}
Local Aggregation algorithm outputs at most $2\sqrt{c{W}}$ MST fragments each of diameter at most $O(\sqrt{W/c})$ in $O(\sqrt{{W/c}} \log^* n)$ rounds and requiring $O(m+ n \log^* n \log (W/c))$ messages.
\end{lemma}
\onlyLong{
\begin{proof}
Each iteration of the local aggregation algorithm performs three major functions, namely finding the MOE, convergecast within the fragment and merging with adjacent fragment over the matched MOE edge. 
For finding the MOE, in each iteration, every node checks each of its neighbor (maximum in $O(D)$ time) in non-decreasing order of weight of the connecting edge starting from the last checked edge (from previous iteration). Thus, each node contacts each of its neighbors at most once, except for the last checked node (which takes one message per iteration). Hence total message complexity (over $\log \sqrt{W/c}$ iterations) is 
\[ \sum_{v\in V}2d(v) + \sum_{i=1}^{\log \sqrt{W/c}} \sum_{v\in V}1 = m+n\left(\frac{1}{2} \log (W/c)\right)\] \[= O(m+ n \log (W/c)), \] where $d(v)$ refers to the degree of a node. 
The fragment leader determines the MOE for the $i^{th}$ iteration, by convergecasting over the fragment, which requires at most $O(2^i)$ rounds since the diameter of any fragment is bounded by $O(2^i)$ (by Lemma \ref{lem:frag_diam}). 
The fragment graph, being a rooted tree, uses a $O(\log^* n)$ round deterministic symmetry-breaking algorithm \cite{matching, Pandurangan:2017:TMD:3055399.3055449} to obtain the required matching edges in the case without latencies.
Taking into account the required scale-up in case of the presence of latencies, the symmetry breaking algorithm
is simulated by the leaders of neighboring fragments by communicating with
each other; since the diameter of each fragment is bounded by $O(2^i)$ and the maximum weight of the MOE edges is also $2^i$, the time
needed to simulate one round of the symmetry breaking algorithm in iteration $i$
is $O(2^i)$ rounds. Also, as only the MST edges (MOE edges) are used in communication, the total
number of messages needed is $O(n)$ per round of simulation. Since there are $\log \sqrt{W/c} =
O(\log (W/c))$ iterations, the total time and message complexity for building the
maximal matching is $O(\sum^{\lceil \log \sqrt{W/c}\rceil}_{i=0} 2^i \log^* n) = O(\sqrt{W/c} \log^* n)$ and $O(n \log^* n)$
respectively. Afterwards, adding selected edges into $M_i'$ (Line $5$ of the local aggregation algorithm) can be done with additional $O(n)$ message complexity and $O(2^i)$ time complexity in iteration $i$. 
Thus, the overall message complexity of the algorithm is
$O(m+ n \log^* n \log (W/c))$ and the overall time complexity is $O(\sqrt{{W/c}} \log^* n)$.
\end{proof}
}

Since there are at most $2\sqrt{{cW}}$ base fragments remaining (from Lemma \ref{lem:CGHS_result}), at most $2\sqrt{c{W}}-1$ MST edges need to be discovered. However, it is at this point where the fashion in which the fragment mergings occur changes. %
From this point on, the base fragments are progressively merged using the shortest path tree $\tau$ \onlyLong{(that was created at the beginning of the algorithm)} in $\log \sqrt{cW} \leq \log n$ iterations.\footnote{If $\sqrt{cW}\geq n$, it implies that only $n$ singleton fragments remain. The given algorithm boils down to the algorithm in Section \ref{subsubsec:ub-unequal}.} 

For easier explanation, we refer to the fragments created in the global aggregation stage as \textit{components}. Initially, a component just constitutes of a single base fragment.

\para{Global Aggregation Stage} (Merging Components using a Shortest Path Tree) %
Each node determines its MOE (w.r.t. its current component), which can be done in parallel, requiring $O(D)$ time and $O(m)$ messages.\footnote{Note that the value of $D$ (in fact a 2-approximation value of $D$) is known through the shortest path tree construction and a node need not wait for more than $2D$ time to determine its MOE as no edge with weight (latency) $>D$ would be present in the MST.} This MOE information is upcast along the fragment edges to the base fragment leader (rather than the component leader), while filtering all but the lightest edge, requiring $O(\sqrt{W/c})$ rounds (base fragment diameter) and $O(n)$ messages (as only lightest MOE is upcast). \onlyLong{(See Figure \ref{fig:Global}).}
Each base fragment leader upcasts the lightest known outgoing edge of the component that it belongs to up  the shortest path tree where intermediate nodes  wait until they receive at least one message from each of its children and then upcast the lightest edge of each component that they have received or belong to (starting from the component with the lowest id) to its parent. After receiving the first message from each child node, subsequent messages arrive in a pipelined order with intervals of $1/c$ (as edges have capacity $c$). As a total of at most $O(\sqrt{cW})$ messages (i.e., one message from each base fragment) is upcast on $\tau$, the maximum possible time required is  $O(D+\frac{1}{c} \sqrt{cW})$. Correspondingly, the number of messages required is $O(D_{SPT}\cdot \sqrt{cW})$, where $D_{SPT}$ is the hop diameter of the shortest path tree $\tau$. As in this case, we assume the latency diameter $D$ to be $\leq \sqrt{W/c}$. The message complexity becomes at most $O(W)$. The root of the shortest path tree locally computes the component mergings (by locally simulating the local aggregation algorithm) and thereafter informs all the fragment leaders of their updated component ids, which further downcasts this to all the nodes; completing an iteration. The guarantee, like earlier, is that the number of components halves in every iteration, requiring a total of at most $\log \sqrt{cW} \leq \log n$ iterations.

\begin{lemma} \label{lem:GA}
Given that there are at most $2\sqrt{c{W}}$ MST fragments after the local aggregation, the global aggregation algorithm outputs the MST requiring $O(D+\sqrt{W/c})$ time and $O(W)$ messages.
\end{lemma}

\onlyLong{
\begin{figure*}
	\centering
	\includegraphics[scale=0.7]{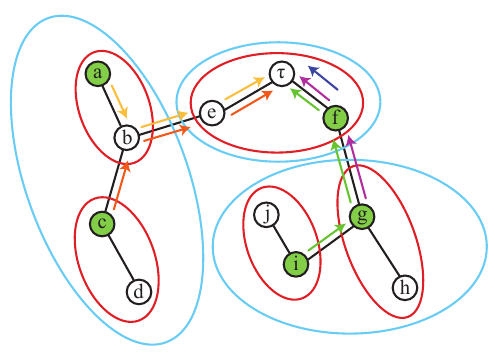}
	\caption[Merging Fragments Using Global Aggregation]{Merging Fragments using Global Aggregation. The figure shows a shortest path tree rooted at $\tau$. The red boundary denotes the nodes within a base fragment and the blue boundary denotes the MST-components that the base fragment belong to. The green nodes are the base fragment leaders that handle all communication with the root (denoted by the colored arrows).}
	\label{fig:Global}
\end{figure*}

}

Next, we consider the case where the latency diameter $D>\sqrt{W/c}$. Here also, we follow a similar procedure of creating a shortest path tree $\tau$ first and thereafter moving on to the local and the global aggregation stages. However in this case we want larger base fragments with a total weight of $O(D)$ (unlike $O(\sqrt{W/c})$ for the previous case) and therefore we run the local aggregation algorithm for $\lceil \log D \rceil$ iterations (instead of $\lceil \log \sqrt{W/c} \rceil$ like the previous case).

\onlyLong{

\begin{algorithm} 
\setlength{\columnsep}{10pt}
\begin{algorithmic}[1]
\State $\mathcal{F}_1 = (V, \phi)$ // initial set consisting of $n$ (singleton) fragments.
\For {$i = 0, \dots , \lceil \log D \rceil$}
\State Each fragment $\mathcal{F} \in \mathcal{F}_i$ of weight at most $2^i$ determines the MOE of $\mathcal{F}$ and if the weight of the MOE edge is $\leq 2^i$, the fragment adds the edge to the candidate set $M_i$.
\State Find a maximal matching $M_i' \subseteq M_i$ in the fragment graph $\mathcal{H}_i = (\mathcal{F}_i , M_i )$.
\State If $\mathcal{F} \in \mathcal{F}_i$ of total weight at most $2^i$ has no incident edge in $M_i'$, it adds its
MOE edge into $M_i'$, iff its MOE edge has weight $\leq 2^i$.
\State  $\mathcal{F}_{i+1}$ is obtained by merging all the fragments along the edges selected in $M_i'$.
\EndFor

\end{algorithmic}
\caption{Local Aggregation (when $D>\sqrt{W/c}$): Outputs at most $2(W/D)$ base fragments of diameter at most $O(D)$.}\label{algo:CGHS2.1}
\end{algorithm}

We note that Lemmas \ref{lem:frag_diam}, \ref{lem:frag_weight}, \ref{lem:frag_number} still hold for this case, and as such the following corollaries follow directly.

\begin{corollary} \label{corr:frag_diam}
At the end of the local aggregation algorithm, each fragment has diameter at most $O(D)$. 
\end{corollary}

\begin{corollary} \label{corr:frag_number}
At the end of the local aggregation algorithm, the number of fragments remaining is at most $2(W/D)$, where $W$ refers to the total weight of the MST and $D$ is the latency diameter.
\end{corollary}
}

\begin{lemma} \label{lem:LA1}
For the case where $D>\sqrt{W/c}$, the local aggregation algorithm outputs at most $2(W/D)$ MST fragments each of diameter at most $O(D)$ in $O(D \log^* n)$ rounds and requiring $O(m+ n\log D \log^* n)$ messages.
\end{lemma}

\onlyLong{
\begin{proof}
As before, each iteration of the local aggregation algorithm performs three major functions, namely finding the MOE, convergecast within the fragment and merging with adjacent fragment over the matched MOE edge. 
For finding the MOE, in each iteration, every node checks each of its neighbor (maximum in $O(D)$ time) in non-decreasing order of weight of the connecting edge starting from the last checked edge (from previous iteration). Thus, each node contacts each of its neighbors at most once, except for the last checked node (which takes one message per iteration). Hence total message complexity (over $\log D$ iterations) is 
\[ \sum_{v\in V}2d(v) + \sum_{i=1}^{\log D} \sum_{v\in V}1 = O(m+ n \log D), \] where $d(v)$ refers to the degree of a node. 
The fragment leader determines the MOE for the $i^{th}$ iteration, by convergecasting over the fragment, which requires at most $O(2^i)$ rounds since the diameter of any fragment is bounded by $O(2^i)$ (by Lemma \ref{lem:frag_diam}). 
The fragment graph, being a rooted tree, uses a $O(\log^* n)$ round deterministic symmetry-breaking algorithm \cite{matching, Pandurangan:2017:TMD:3055399.3055449} to obtain the required matching edges in the case without latencies.
Taking into account the required scale-up in case of the presence of latencies, the symmetry breaking algorithm
is simulated by the leaders of neighboring fragments by communicating with
each other; since the diameter of each fragment is bounded by $O(2^i)$ and the maximum weight of the MOE edges is also $2^i$, the time
needed to simulate one round of the symmetry breaking algorithm in iteration $i$
is $O(2^i)$ rounds. Also, as only the MST edges (MOE edges) are used in communication, the total
number of messages needed is $O(n)$ per round of simulation. Since there are $
O(\log D)$ iterations, the total time and message complexity for building the
maximal matching is $O(\sum^{ \log D}_{i=0} 2^i \log^* n) = O(D \log^* n)$ and $O(n \log^* n)$
respectively. Afterwards, adding selected edges into $M_i'$ (Line $5$ of the local aggregation algorithm) can be done with additional $O(n)$ message complexity and $O(2^i)$ time complexity in iteration $i$. 
Thus, the overall message complexity of the algorithm is
$O(m+ n\log D \log^* n)$ and the overall time complexity is $O(D \log^* n)$.
\end{proof}

}

The global aggregation algorithm \onlyLong{(for the case where $D>\sqrt{W/c}$)} is run as such without any change over these larger base fragments. However, as the base fragments here have a total weight of at most $D$, the number of base fragments for this case is $O(W/D)$. Each of $O(W/D)$ base fragment upcasts the MOE (w.r.t. its current component) to the root of the shortest path tree $\tau$ requiring at most $O(D+ \frac{1}{c}W/D) = O(D+\sqrt{W/c})= O(D)$ rounds (as for this case $D>\sqrt{W/c}$). Correspondingly, the number of messages required is $O(D_{SPT}\cdot W/D)$, where $D_{SPT}$ is the hop diameter of the shortest path tree $\tau$. As $D_{SPT} \leq D$, the message complexity becomes at most $O(W)$.

\begin{lemma} \label{lem:GA1}
For the case where $D>\sqrt{W/c}$, given that the local aggregation algorithm outputs at most $2(W/D)$ MST fragments, the global aggregation algorithm outputs the MST requiring $O(D)$ time and $O(W)$ messages.
\end{lemma}

\onlyLong{
The overall time and message complexity is determined by the cost of local aggregation algorithm along with the cost of pipelining over the shortest path tree. 
Combining the above results, we obtain a time complexity of  $\Tilde{O}(D+\sqrt{W/c})$ and the message complexity is $\Tilde{O}(m+ W)$.
}

\para{Guessing and Doubling}
Without the knowledge of the total weight of the MST, we can still perform the above algorithm, by guessing a value for $W$ in each iteration. Starting with an initial guess of $1$, if the algorithm succeeds for the guessed value \onlyLong{ of $W$}, it terminates; otherwise it doubles the guessed value and continues. 

First, a shortest path tree $\tau$ is built. Thereafter, local aggregation is performed with the guessed value of $W$. To check for success, each base fragment leader sends a single bit to the root of the shortest path tree such that the root can determine the count of the total number of base fragments present. Depending on the relationship between the latency diameter $D$ and the current estimate of $W$, for the case where $D \leq \sqrt{W/c}$ if the number of base fragments is $\leq 2\sqrt{c{W}}$ (c.f. Lemma \ref{lem:CGHS_result}) or for the case where $D > \sqrt{W/c}$  if the number of base fragments is $\leq 2(W/D)$ (c.f. Lemma \ref{lem:LA1}), it would imply that the algorithm was successful at guessing the value of $W$. %
Once the root determines the appropriate value of $W$, it intimates all the other nodes to run the actual algorithm. %
Note that, this can add up to a factor of $\log W$ to obtain the correct guess.

\begin{theorem}
In the weighted CONGEST model, when edge latencies equal edge weights there exists an algorithm that computes the MST in $O((D+\sqrt{W/c})\log^* (n) \cdot \log W)$ rounds and using $O((m+W+ n \log^* n (\log D + \log (W/c)))\log W )$ messages w.h.p.
\end{theorem}

\onlyLong{
\begin{proof}
The correctness of the algorithm immediately follows from the fact that in each iteration MST fragments merge with one another, and the total number of fragments reduce by at least half. This property (or invariant) ensures that at the end there is only one fragment remaining, which is in fact the MST.	

Creating the shortest path tree takes $O(D)$ time (c.f. Lemma \ref{lem:spt}). Thereafter, the overall time and message complexity is determined by the cost of local aggregation algorithm,  the cost of merging the components over the shortest path tree (i.e. the global aggregation algorithm) for each case based on the latency diameter $D$, along with the cost for guessing the correct $W$.\\
\textbf{Case 1: When} $D\leq \sqrt{W/c}$\\ 
Combining the time complexity for both local and global aggregation (given by Lemmas \ref{lem:CGHS_result} and \ref{lem:GA}), we obtain an overall time complexity of $O(\sqrt{{W/c}} \log^* n) + O(D+\sqrt{W/c}) = O(\sqrt{{W/c}} \log^* n)$ rounds (as $D\leq \sqrt{W/c}$). Similarly, the overall message complexity is $O(m+W+ n \log^* n \log (W/c))$.\\
\textbf{Case 2: When} $D> \sqrt{W/c}$\\ 
Combining the time complexity for both local and global aggregation (given by Lemmas \ref{lem:LA1} and \ref{lem:GA1}), we obtain an overall time complexity of $O(D\log^* n) + O(D) = O(D \log^* n)$ rounds. Similarly, the overall message complexity is $O(m+W+ n \log^* n \log D)$.

Combining both cases, the time complexity of the algorithm (given the value of $W$ is known) is given by  $O((D+\sqrt{W/c})\log^* n)$ rounds  and the message complexity is given by  $O(m+W+ n \log^* n (\log D + \log (W/c)) )$. 

Using the guess and double technique requires $\log W$ iterations which increases the overall time and message complexities by a factor of $\log W$. Therefore, the overall time and message complexities are given by $O((D+\sqrt{W/c})\log^* n \log W)$ rounds and $O((m+W+ n \log^* n (\log D + \log (W/c)))\log W )$ respectively.
\end{proof}
}

\enlargethispage{\baselineskip}
\subsection{Lower Bounds} \label{subsec:lbcorrelated}
In this section, we provide a time complexity lower bound of $\tilde{\Omega}(D + {\sqrt{W/c}})$  
for the case where edge weights also represent edge latencies, through the following theorem. 
For message complexity, based directly from the results of \cite{Kutten:2015:CUL:2742144.2699440}, we show a lower bound of $\Omega(m)$. We further show that any algorithm that runs in a constant factor of the optimal running time, for a particular choice of constant, requires $\Omega(W)$ messages in the worst case.

\begin{theorem} \label{thm:lb-l=w}
Any algorithm to compute the MST of a network graph in which the weights correspond to latencies must, in the worst case, take $\Omega(D + ({\sqrt{W/c}})/{\log n})$ time, where $W$ is the total weight of the MST. 
\end{theorem}
\begin{proof}

First, we show a lower bound of $\Omega(D)$. Consider the network that is a ring with $n-2$  edges of latency (and therefore weight) $\ell$  and the remaining two edges $e_1$ and $e_2$ are  positioned diametrically opposed to each other and are assigned latencies $\ell+r_1$ and $\ell+r_2$ resp., where $r_1$ and $r_2$ are random integers from, say, $[1, \ell]$. Clearly, any MST algorithm must take $\Omega(n\ell) = \Omega(D)$ time to determine whether $e_1$ or $e_2$ must be in the MST. The bound holds for all values of $D$ as long as $n \ge 4$ and $\ell$ can be suitably adjusted to be $\Omega(D/n)$.

Next, we show a lower bound of $\Omega(({\sqrt{W/c}})/{\log n})$ in two steps. We first relate algorithms in our model to algorithms in the classical CONGEST model (cf. Lemma~\ref{lem:simulation}). Then, we complete the lower bounding argument by applying the lower bound from Das Sarma et al.~\cite{DHK+12}.

\begin{lemma} \label{lem:simulation}
  Assume that we are given an algorithm $\mathcal{A}$ in the weighted
  CONGEST model and that $\mathcal{A}$ runs in $T$ rounds on a given
  graph $G$ with minimum edge latency $\ell_{\min}$, maximum edge
  latency $\ell_{\max}$ and capacity $c$. Then, $\mathcal{A}$
  can be run in $O(T/\ell_{\min} + 1)$
  rounds in the standard CONGEST model with messages of size
  $O(c \cdot \ell_{\min}\cdot \log n)$ bits. 
\end{lemma}
\begin{proof}
  We first convert algorithm $\mathcal{A}$ into version $\mathcal{A}'$
  that is slowed down by an integer factor
  $\alpha := 2\cdot \lceil\ell_{\max}/\ell_{\min}\rceil$. In the
  converted algorithm,
  messages are only sent at times that are integer multiples of
  $\alpha$. If a message is supposed to be sent at time $t$ in
  algorithm $\mathcal{A}$, we send the message at time $\alpha\cdot t$
  in $\mathcal{A}'$. Note that if the running time of
  $\mathcal{A}'$ is $T$, the running time of $\mathcal{A}'$ is
  $O(\alpha T)$.

  We first show that by doing this scaling of the algorithm, each node
  already knows what messages it sends at time $t$ in algorithm
  $\mathcal{A}'$ quite a bit before time $t$. Consider a node $v$ and
  some message $M$ that is sent by node $v$ at time $t$ in algorithm
  $\mathcal{A}$ and thus at time
  $\alpha\cdot t$ in algorithm $\mathcal{A}'$. In $\mathcal{A}$,
  $v$ knows the messages it sends at time $t$ after receiving all the
  messages that are received by $v$ by time $t$. Consider a message $M'$ that is
  received by $v$ from some neighbor $u$ at time $t_v\leq t$ in
  $\mathcal{A}$.  Let $\ell\in [\ell_{\min},\ell_{\max}]$ be the
  latency of the edge $\{u,v\}$. In $\mathcal{A}$, $u$ sends the
  message at a time $t_u\leq t_v-\ell$. In the slowed down algorithm
  $\mathcal{A}'$, $u$ therefore sends the message at the latest at
  time $\alpha(t_v-\ell) \leq \alpha t - \alpha\ell$. Because the
  latency of the edge is $\ell$, the message is thus received by $v$
  at the latest at time $t-(\alpha-1)\ell$. Node $v$ therefore knows
  all the information required for the messages it sends at  time
  $\alpha t$ in $\mathcal{A}'$ already at least $(\alpha-1)\ell \geq
  (\alpha-1)\ell_{\min}$ time units prior to sending the message.
  Because $\alpha\geq 2\ell_{\max}/\ell_{\min}$, in $\mathcal{A}'$,
  all nodes thus know which messages to send at a given time $t$ at
  least $(\alpha-1)\ell_{\min}\geq \ell_{\max}$ rounds prior to time
  $t$.

  As a next step, we convert $\mathcal{A}'$ to an algorithm
  $\mathcal{A}''$ where nodes can send larger messages, but where all
  messages are only sent at times that are integer multiples of
  $\ell_{\max}$. Because messages in $\mathcal{A}'$ are known at least
  $\ell_{\max}$ time units prior to being sent, a message that is sent by
  $\mathcal{A}'$ at time $t$ can be sent by $\mathcal{A}''$ at the latest time $t'\leq
  t$ such that $t'$ is an integer multiple of $\ell_{\max}$. Note that
  because all messages in $\mathcal{A}''$ are sent at the latest at
  the time when they are sent by $\mathcal{A}'$, all messages are
  available in $\mathcal{A}''$ when they need to be sent and the time
  complexity of $\mathcal{A}''$ is at most the time complexity of
  $\mathcal{A}'$.  The number of messages of $\mathcal{A}'$ that have to be combined
  into a single message of $\mathcal{A}''$ is at most the number of
  messages that are sent over an edge in an interval of length
  $\ell_{\max}$ by $\mathcal{A}'$ and thus in an interval of length
  $\lceil\ell_{\max}/\alpha\rceil$ by the original algorithm
  $\mathcal{A}$. Because the capacity of each edge is most $c$, the
  number of messages that have to be combined into a single message of
  $\mathcal{A}''$ is thus at most
  $O(c \ell_{\max}/\alpha) = O(c \ell_{\min})$.

  Because $\mathcal{A}''$ only sends messages at times that are
  integer multiples of $\ell_{\max}$, the algorithm also works if we
  increase the latency of each edge to be $\ell_{\max}$. This might
  increase the total time complexity by one additive $\ell_{\max}$
  because at the end of the algorithm, the nodes might have to receive
  the last message before computing their outputs. If $T$ is the time
  complexity of $\mathcal{A}$, the time complexity of this modified
  $\mathcal{A}''$ is therefore at most $O(\alpha\cdot T +
  \ell_{\max})$. If all the edge
  latencies are integer multiples of $\ell_{\max}$, the model is
  exactly equivalent to the original CONGEST model with messages of
  size $O(c\ell_{\min}\log n)$ bits, where time is scaled by a factor
  $\ell_{\max}$ and the claim of the lemma thus follows. %
\end{proof}
From Das Sarma et al.~\cite{DHK+12}, we know that computing the MST in the $\mathsf{CONGEST}$ model with bandwidth $O(c\ell_{\min}\log n)$ bits requires $\Omega(\sqrt{{n}/{B \log n}})$ rounds; here $B$ is the bandwidth term referring to the number of bits that can be sent over an edge per round. Note that their construction uses edge weights 0 and 1, but since the MST will be the same when edge weights are offset by 1, their lower bound holds for the case where their edge weights 0 and 1 are changed to 1 and 2, respectively.  This, in turn, translates to a lower bound of $\Omega(\sqrt{{n \ell_{\min}}/{c \log^2 n}}) = \Omega(({\sqrt{(W/c)}})/{\log n})$.
\end{proof}

In regard to the message complexity, we show a lower bound of $\Omega(m)$ through a `proof by contradiction' that follows directly from the results in \cite{Kutten:2015:CUL:2742144.2699440} by considering edges to have unit weights and latencies. If it was possible to construct an MST with $o(m)$ messages, then we could solve single-source broadcast by using only the MST edges, contradicting the lower bound in \cite{Kutten:2015:CUL:2742144.2699440} (Corollary~3.12).

\begin{figure}[htb]
	\begin{center}
		\includegraphics[page=11,scale=0.5,clip,trim=3.3in 0.1in 3.6in 3.2in]{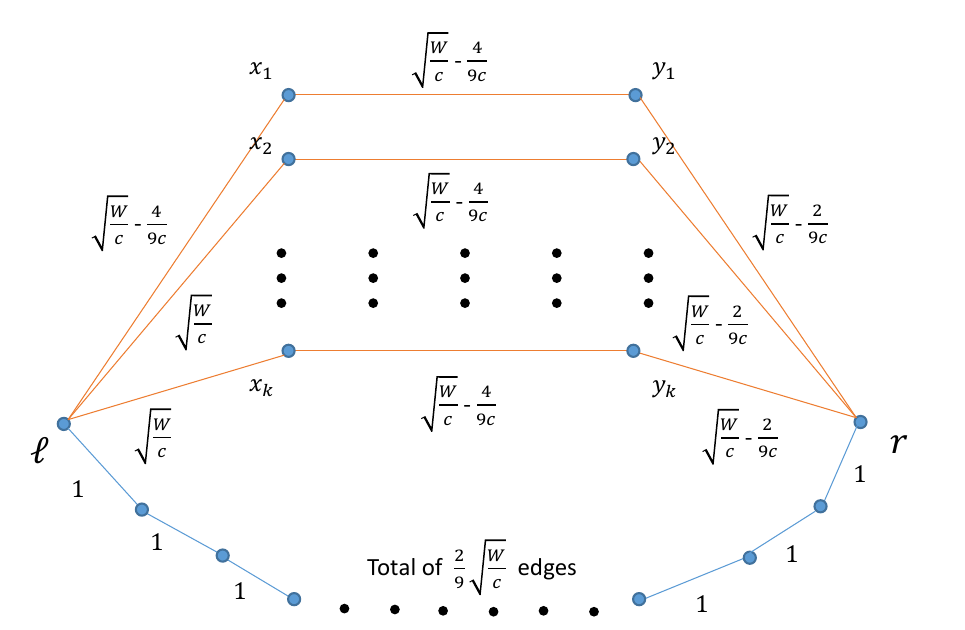}
	\end{center}
	\caption{Construction for showing message complexity lower bound when edge weights equal edge latencies.}
	\label{fig:lbz}
\end{figure}

Next, we present a construction (cf. Figure~\ref{fig:lbz}) that shows a $\Omega(W)$ lower bound on the message complexity for any algorithm computing MST in time  $\leq ({5}/{27})(D+\sqrt{{W}/{c}})$ %
rounds (for the case when edge weights equals edge latencies). The construction comprises two nodes that we designate $\ell$ and $r$ connected by a long path of $({2}/{9})\sqrt{{W}/{c}}$ %
edges, each of weight $1$. All edges uniformly have a capacity $c$. We then add $k = \sqrt{cW}/3$ parallel paths between $\ell$ and $r$, with each $i$th path comprising a left edge $(\ell, x_i)$, a middle edge $(x_i, y_i)$, and a right edge $(y_i, r)$. Each of the left edges has weight either $\sqrt{W/c}$ or $\sqrt{{W}/{c}}-({4}/{9c})$ %
chosen uniformly and independently at random (u.i.r.), while all the middle and the right edges have weights of $\sqrt{{W}/{c}}-({4}/{9c})$ and $\sqrt{{W}/{c}}-({2}/{9c})$, respectively. %
Note that the sum of all the edge weights (in expectation), equals W. Furthermore, the latency diameter $D= \sqrt{W/c} + (2/9)\sqrt{W/c} + \sqrt{{W}/{c}}-({2}/{9c}) = $ $(20/9)\sqrt{{W}/{c}}-({2}/{9c})$, is determined by the distance between some $x_i$ (that is connected to $\ell$ with an edge of latency $\sqrt{W/c}$) and $y_j$ such that $i \neq j$. 

It is clear that for each path, the middle edge will always be included in the MST, but we must include (exclusively) either the left edge or the right edge in the MST. Specifically, $r$ must know which of its $k$ incident right edges must be included in the MST. Notice that there are $2^k$ equally likely possibilities from which $r$ must compute the correct outcome. This requires $r$ to learn a random variable $X$ that encodes these $2^k$ possibilities. Recall that the Shannon's entropy of $X$ given by $H(X) = - \sum_{i=1}^{2^k} (1/2^k) \log(1/2^k) = k$ is a lower bound on the number of bits that $r$ must learn. By design, for any $\sqrt{cW}\geq 1$, none of the middle edges can be used for communication as the latency of the middle edge is be greater than the allowable time complexity of $({5}/{27})(D+\sqrt{{W}/{c}})$  %
rounds. Therefore, these $k$ bits must be learned exclusively through the long path of $({2}/{9})\sqrt{{W}/{c}}$ %
edges. This would require $(({2}/{9})\sqrt{{W}/{c}})+({1}/{c})({\sqrt{cW}}/{3}) = ({5}/{9})\sqrt{{W}/{c}}$ %
rounds (which is $\leq$ the required time complexity of $({5}/{27})(D+\sqrt{{W}/{c}})$) %
and would require $(({2}/{9})\sqrt{{W}/{c}})\cdot ({\sqrt{cW}}/{3}) = 2W/27 = \Omega(W)$  %
messages. Suppose $r$ can learn the MST in  $o\left(W\right)$ messages. Then, it would imply that fewer than $k$ bits were required to learn $X$, a contradiction. %

\begin{theorem}\label{thm:lbz}
For the case where edge weights equal edge latencies, any MST construction algorithm in the weighted $\mathsf{CONGEST}$ model that runs in time $\leq \frac{5}{27}\left(D+\sqrt{\frac{W}{c}}\right)$ rounds requires $\Omega(W)$ messages in the worst case, given that the assumption of $\sqrt{cW}\geq 1$ holds. 
\end{theorem}

\section{Unrelated  Weights and  Latencies} \label{sec:uncorrelated}
In this section, we consider the case when there is no relationship between the edge weights and the latencies, and either can take arbitrary values. We show that unlike the $\Tilde{\Theta}(D'+\sqrt{n})$ rounds tight bounds for MST construction in the standard \textsf{CONGEST} model (where $D'$ refers to the  diameter with unit latencies), the best that can be achieved in this case is $\Tilde{\Theta}(D+{n/c})$. %

\subsection{Lower Bound} 
We now present a construction (cf. Figure~\ref{fig:lb1}) that shows a $\Omega(n/c)$ lower bound on the time complexity for computing MST when edge weights are independent of latencies. The construction is similar in flavour to the message complexity lower bound given in the previous section. As before, the lower bound graph comprises two nodes designated as $\ell$ and $r$, that are connected here by an edge $(\ell,r)$ of weight $1$. We then add $k = n/2 - 1$ parallel paths between $\ell$ and $r$, with each $i$th path comprising a left edge $(\ell, x_i)$, a middle edge $(x_i, y_i)$, and a right edge $(y_i, r)$. Each of the left  edges has weight either 1 or 3 chosen %
uniformly and independently at random, while all the middle and right edges have weights 1 and 2, respectively.  All the middle edges, in this case, have very high latencies of the form $\omega(n/c)$, whereas all other edges have a latency 1. All edges uniformly have a capacity $c$. 
Basically, due to its low weight, the middle edge will always be included in the MST, but we must include (exclusively) either the left edge or the right edge in the MST. Specifically, $r$ must know which of its $k$ incident right edges must be included in the MST. \onlyShort{Using Shannon's entropy function, as done in the previous section, we show that $k$ is a lower bound on the number of bits that $r$ must learn.} \onlyLong{Notice that there are $2^k$ equally likely possibilities from which $r$ must compute the correct outcome. This requires $r$ to learn a random variable $X$ that encodes these $2^k$ possibilities. Recall that the Shannon's entropy of $X$ given by $H(X) = - \sum_{i=1}^{2^k} (1/2^k) \log(1/2^k) = k$ is a lower bound on the number of bits that $r$ must learn.}  These $k$ bits must be learned exclusively through the edge $(\ell, r)$. Suppose $r$ can learn the MST in  $o\left(\frac{n}{c \log n}\right)$ rounds. Then, this algorithm can be used to learn $X$ with fewer than $k$ bits which leads to a contradiction.

\onlyShort{
\begin{figure}[htb]
	\begin{center}
		\includegraphics[page=12,scale=0.4,clip,trim=3.3in 1.2in 4in 3.4in]{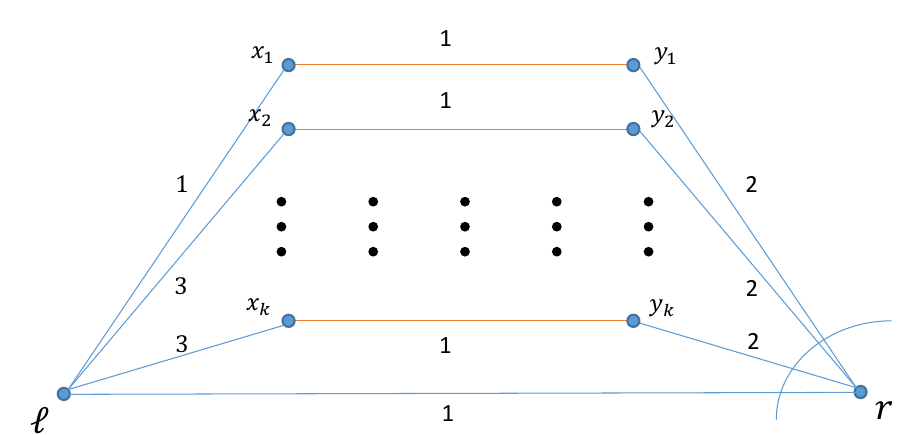}
	\end{center}
	\caption{Construction for showing time complexity lower bound when edge weights are different from edge latencies.}
	\label{fig:lb1}
\end{figure}
}

\onlyLong{
\begin{figure}[htb]
	\begin{center}
		\includegraphics[page=9,scale=0.55,clip,trim=3.3in 1.2in 4in 1in]{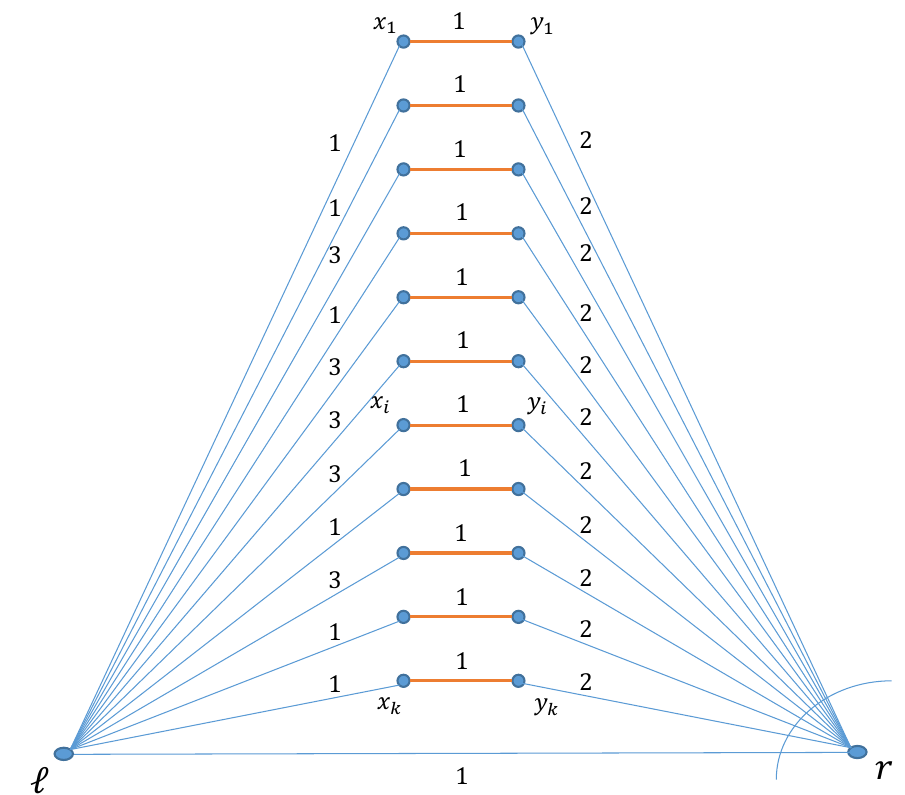}
	\end{center}
	\caption{Construction for showing lower bound when edge weights are different from edge latencies.}
	\label{fig:lb1}
\end{figure}
}

The lower bound $\Omega(D)$  has already been addressed in Section~\ref{subsec:lbcorrelated} (in the proof of Theorem \ref{thm:lb-l=w}).

This implies the following result:
\begin{theorem}
Any algorithm (deterministic or randomized) for computing  the MST of a network in which the edge weights and the latencies are independent of each other requires $\tilde{\Omega}(D+ n/c)$ time.
\end{theorem}

\subsection{Upper Bound} \label{subsubsec:ub-unequal}

Here, we provide an ${O}(D+n/c)$ time algorithm for constructing an MST when there is no relationship between an edge's latency and its weight. The algorithm is based on the pipeline algorithm \cite{Pipeline,Gopalsurvey} for the standard \textsf{CONGEST} model. 
The key idea here is to create a shortest path tree where nodes upcast information, starting from the leaf nodes while filtering non-essential information. The root computes the MST and broadcasts it to all the nodes.

\onlyLong{The basic outline of our algorithm is as follows. 
First, a particular node elects itself as the leader. Next, a shortest path tree w.r.t. latency is created with the leader as the root node. 
Nodes upcast information, starting from the leaf nodes while filtering non-essential information. The root computes the MST and broadcasts it to all the nodes.}

\onlyShort{\para{MST Algorithm for Arbitrary Weights and Latencies}}
To determine the MST, the leaf nodes of the shortest path tree start by upcasting its adjacent edges in non-decreasing order of weight. Intermediate nodes begin only after having received at least one message from each of its children in the shortest path tree. From the set of all the edges received until the current round (along with its own adjacent edges) intermediate nodes filter and upcast only the lightest $n-1$ edges that do not create a cycle. %
Notice that for any intermediate node, after it receives the first message from all of its children, all subsequent messages (at most $n-1$ from each child) arrive in a pipelined manner with an interval of $1/c$ (as edges have capacity $c$). Moreover, as the intermediate nodes start upcasting immediately after receiving the first message from each of its children, they also send at most $n-1$ messages up in a pipelined fashion, while filtering the heavier cycle edges. \onlyLong{Waiting for at least one message from each child in the shortest path tree and the fact that messages are always upcast in a non-decreasing order ensures that in every round (after receiving at least one message from each child), nodes have sufficient data to upcast the lightest edge.}

To identify edges that form a cycle, all nodes excepting the root, maintain two edge lists, $Q$ and $U$. Initially, for a vertex $v$, $Q_v$ contains all the edges adjacent to $v$, and $U_v$ is empty. At the time of upcast, $v$ determines the minimum-weight edge set in $Q_v$ that does not create any cycle with the edges in $U_v$ and upcasts it to its parent while moving this edge from $Q_v$ to $U_v$. Every parent node adds all the messages received from its children to its $Q$ list. Finally, $v$ sends a terminate message to its parent when $Q_v$ is empty. This filtering guarantees that each node upcasts at most $n-1$ edges to its parent.
As edges have capacity $c$, this requires at most $O(n/c)$ rounds. Considering any path of the shortest path tree from a leaf \onlyLong{node} to the root, the maximum number of messages that are sent in parallel at any point of time on this path is at most $n-1$. Since messages are \onlyLong{always} upcast in a pipelined fashion  the time complexity is $O(D+{n}/{c})$ rounds. As each node sends at most $n-1$ messages, the message complexity is $O(n^2)$. 
Thus, we have shown the following result:

\begin{theorem}
In the weighted CONGEST model, there exists an algorithm that computes the MST in $O(D + {n}/{c})$ rounds and with $O(m\log n + n^2)$ messages w.h.p.
\end{theorem}

\onlyLong{
\begin{proof}

The algorithm's correctness follows directly from the cycle property \cite{tarjan, Kleinberg:2005:AD:1051910}. The filtering rule of the algorithm ensures that any edge sent upward by a node $v$ does not close a cycle with the already sent edges (edges in list $U_v$). Since the edges are upcast in a non-decreasing order of weight, and intermediate nodes begin only after receiving at least one message from each of its neighbors, in ensures that each intermediate node has enough information to send the correct lightest edge of say weight $w$. This implies that, in no later round does that intermediate node receive a message of weight less than $w$. As such, the only edges filtered are the heaviest cycle edges, which implies that none of the MST edges are ever filtered. The root receives all the MST edges (and possibly additional edges) required to compute the MST correctly. Termination is guaranteed as each node sends at most $n-1$ edges upwards and then a termination message. For a more detailed correctness proof, refer to the proofs in \cite{Pipeline} and \cite{Gopalbook}. %

The time complexity is determined by the cost of creating the shortest path tree and the cost of doing a pipelined convergecast on this tree. 
The creation of the shortest path tree requires $O(D)$ time and $O(m\log n)$ messages (c.f. Lemma \ref{lem:spt}).
The pipelined convergecast is started by the leaf nodes by sending their lightest $n-1$ adjacent edges up.
Thereafter each intermediate node upcasts only after receiving at least one message from each of its children, and this upcast of lightest $n-1$ edges that does not create a cycle happens in a pipelined fashion. The maximum delay at any intermediate node would be due to waiting for the messages from its furthest sub-tree node. From the definition of the graph diameter, this delay is bounded by $D$. This implies that in the absence of congestion, the root node would receive all the required information in $O(D)$ time. Secondly, from the filtering (and cycle property), it is guaranteed that the congestion at any point is not more than $n-1$. As all edges in a path have capacity $c$, the delay due to congestion is at most $O(n/c)$. (Similar for root node broadcasting the MST over the shortest path tree). %
Therefore, the  total time complexity of weighted pipeline algorithm is $O(D+(n/c))$ (combining the cost of creating a shortest path tree with the cost of congestion). Furthermore, as each node can send at most $n-1$ edges, the message complexity is bounded by $O(m\log n + n^2)$.
\end{proof}
}

\section{Uniform Latencies, Different Weights} \label{subsubsec:ub-equal}

In this section, we consider all edges to have the exact same latency $\ell$, while each edge can have an arbitrarily different weight\onlyShort{.} \onlyLong{(ties can be broken using the node ids). It is here, where we emphasize the role of edge capacities in getting a faster solution. }Given that all edges have the same latency $\ell$, one would expect an $O(\ell)$ slowdown from the results for the standard unit-latency model. This is in fact true, if we consider the worst case capacity $c=1/\ell$, %
where a message can be sent over an edge only after the previous message has been delivered. However, for the case $c=1$, %
where a new message can be sent over an edge in every round; instead of a direct multiplicative factor slowdown to $\Tilde{O}(\ell(D'+\sqrt{n}))$ (where $D'$ is the graph diameter with unit latencies), we obtain an upper bound of $\Tilde{O}(D+\sqrt{n\ell})$ by exploiting the additional capacity by pipelining messages. 
Note, when all edge latencies are the same, then the latency diameter $D=\ell\cdot D'$. More generally, for any given latency $\ell \geq 1/c$, we give an algorithm that constructs an MST in $\Tilde{O}(D+\sqrt{{n\ell}/{c}})$ time.
This section best illustrates the power of having a larger edge capacity, which our algorithm leverages when pipelining messages over an edge.

\subsection{MST Algorithm for Uniform Latencies} 
To obtain an algorithm that is simultaneously both time and message optimal, we base our algorithm on Elkin's MST algorithm \cite{Elkin:2017:SDD:3087801.3087823} and the algorithm in Section \ref{subsubsec:ubcorrelated}. %
As earlier, the algorithm is primarily divided into two stages, the local and the global aggregation.

In order to obtain optimal time complexity, the key lies in determining the ideal switching point between stages and thereby obtaining a correct balance between the costs of local and global aggregation. With the presence of latencies and capacities, naively executing any unit latency MST construction algorithm (\cite{Elkin:2017:SDD:3087801.3087823}, \cite{Pandurangan:2017:TMD:3055399.3055449}) does not lead to this ideal balance point. Rather than only depending on the number of nodes $n$, this balance point %
now depend on the latency, capacity as well the number of nodes $n$. In this case, where all edges have latency $\ell$ and capacity $c$, the balance between the costs of local and global aggregation is achieved when the base fragment diameter equals $\sqrt{n\ell/c}$.

However, in order to obtain optimal message complexity, the algorithm distinguishes between two cases based on the latency diameter. If $D\leq \sqrt{n\ell/c}$, we build base fragments of diameter $\sqrt{n\ell/c}$ in the local aggregation phase, whereas when $D>\sqrt{n\ell/c}$ the base fragments are built with fragment diameter of $D$. As all edges have uniform latency, global aggregation is done using a BFS tree (rather than a shortest path tree) that is constructed using the BFS tree construction algorithm for the standard \textsf{CONGEST} model in $O(D)$ time and with $O(m)$ messages \cite{Elkin:2017:SDD:3087801.3087823}. Here also, the guarantee is that the number of fragments/components halves in each iteration.
We observe that, this careful determination of the parameter $\sqrt{n\ell/c}$ results in a speed-up (w.r.t. the expected running time of $\Tilde{O}(\ell(D'+\sqrt{n}))$, where $D'$ is the graph diameter with unit latencies) by taking advantage of the edge capacities for possible pipelining of messages.

\onlyShort{For the exact details of the MST algorithm for the uniform latency case, refer to the full version of the paper \cite{augustine2019latency}.
We summarize the result via the following theorem.}

\onlyLong{
\para{Algorithm} The algorithm begins by creating a BFS tree. Since all edges have uniform latency, we can construct a BFS tree using the BFS tree construction algorithm for the standard \textsf{CONGEST} model in $O(D)$ time and with $O(m)$ messages \cite{Elkin:2017:SDD:3087801.3087823}.
This can be easily done by scaling one round of the standard \textsf{CONGEST} model to $\ell$ rounds here. The time taken will be given by $O(\ell D')$ which is equal to $O(D)$. 

\para{Local Aggregation Stage} 
Local Aggregation begins with each node as a singleton fragment and thereafter in every iteration, fragments merge in a controlled and balanced manner, while ensuring that the number of fragments at least halve in each iteration. %
When $D\leq \sqrt{n\ell/c}$, we start by building base fragments of diameter $\sqrt{n\ell/c}$, whereas when $D>\sqrt{n\ell/c}$ the base fragments are built with fragment diameter of $D$.

In fact, we show that the total number of fragments that remain after the local aggregation part is $O(\sqrt{{cn\ell}})$, and the diameter of each fragment is at most $O(\sqrt{n\ell/c})$.%
}

\onlyLong{
Similar to the analysis of the algorithm in section \ref{subsubsec:ubcorrelated}, we define the set of fragments $\mathcal{F}_i$, the fragment graph $\mathcal{H}_i$, and the edge set $M_i'$. %
The pseudocode for building base fragments (when $D\leq \sqrt{n\ell/c}$) is shown in Algorithm \ref{algo:MSTF}  and uses similar techniques as in Section \ref{subsubsec:ubcorrelated} and the controlled-GHS algorithm of \cite{Gopalsurvey}.
}

\onlyLong{
\begin{algorithm} 
\setlength{\columnsep}{10pt}
\begin{algorithmic}[1]
\State $\mathcal{F}_1 = (V, \phi)$ // initial set consisting of $n$ (singleton) fragments.
\For {$i = 0, \dots , \lceil \log \sqrt{{n}/{c\ell}} \rceil$}
\State Each fragment $\mathcal{F} \in \mathcal{F}_i$ of diameter $\leq 2^i\ell$ determines the MOE of $\mathcal{F}$ and adds it to candidate set $M_i$.
\State Find a maximal matching $M_i' \subseteq M_i$ in the fragment graph $\mathcal{H}_i = (\mathcal{F}_i , M_i )$.	
\State If $\mathcal{F} \in \mathcal{F}_i$ of diameter at most $2^i\ell$ has no incident edge in $M_i'$, it adds its	
MOE edge into $M_i'$.	
\State  $\mathcal{F}_{i+1}$ is obtained by merging all the fragments along the edges selected in $M_i'$.	
\EndFor	

 \end{algorithmic}	
\caption{Uniform Local Aggregation: Outputs at most $\sqrt{c{n\ell}}$ base fragments of diameter $O(\sqrt{n\ell/c})$.}\label{algo:MSTF}	
\end{algorithm}	
}

 \onlyLong{

 \begin{lemma} \label{lem:fragdiam}	
At the start of the $i^{th}$ iteration, each fragment has diameter at most $O(2^i \ell)$. Specifically, at the end of the uniform local aggregation algorithm each fragment has diameter at most $O(\sqrt{n\ell/c})$. (hop diameter $O(\sqrt{{n}/{c\ell}})$)	
\end{lemma}	

 \begin{proof}	
We show via an induction on the iteration number $i$, that, at the start of the $i^{th}$ iteration, the diameter of each fragment is at most $5 \cdot 2^i \ell$.	
The base case, i.e., at the beginning of iteration $0$, the statement is trivially true, since $5 \cdot 2^i\ell = 5 \cdot 2^0\ell = 5\ell$	
which is greater than $0$, the diameter of a singleton fragment.	
For the induction hypothesis, assume that the diameter of each fragment at the start of	
the the $i^{th}$ iteration is at most $5 \cdot 2^i\ell$. We show that when the $i^{th}$ iteration ends  (i.e., at the start	
of iteration $i + 1$), the diameter of each fragment is at most $5 \cdot 2^{i+1}\ell$.	
We see that a fragment grows via merging with other fragments over a matching edge of the fragment graph. Also, from the algorithm, it is to be noted that	
at least one of the fragments that is taking part in the merging has diameter 	
at most $2^i\ell$ since only fragments with diameter at most $2^i\ell$ find MOE edges; the	
MOE edge may lead to a fragment with larger diameter, i.e., at most $5 \cdot 2^i\ell$. 	

 Additionally, some other fragments (with diameter at most $2^i\ell$) can possibly join with either merging fragments of the matching edge, if they did not have any adjacent matching edge (see Line $5$ of Algorithm \ref{algo:MSTF}). 	

 Therefore, the resulting diameter of the newly merged fragment at the end of	
iteration $i$ is at most $5 \cdot 2^i\ell + 3 \cdot 2^i\ell + 3\ell$, since the diameter of the combined fragment	
is determined by at most $4$ fragments, out of which at most one has diameter	
$5 \cdot 2^i\ell$ and the other three have diameter at most $2^i\ell$ and these are joined by $3$ MOE edges	
(which contributes to $3\ell$). Thus, the diameter at the end of iteration $i$	
is at most $5 \cdot 2^i\ell + 3 \cdot 2^i\ell + 3\ell \leq 8 \cdot 2^i\ell + 3\ell \leq 5 \cdot 2^{i+1}\ell$, for $i \geq 1$. 	

 Since the uniform local aggregation algorithm runs for $\lceil \log \sqrt{{n}/{c\ell}}\rceil$ iterations, the diameter of each fragment at the end of the algorithm is at most $O(2^{\lceil \log \sqrt{{n}/{c\ell}}\rceil}\ell) = O(\sqrt{{n}/{c\ell}}\cdot \ell) = O(\sqrt{n\ell/c})$.	
\end{proof}	

 \begin{lemma} \label{lem:fragsize}	
At the start of the $i^{th}$ iteration, each fragment has size at least $2^i$ and the number of fragments remaining is	
at most $n/2^i$. Specifically, at the end of the uniform local aggregation algorithm, the number of fragments remaining is at most $\sqrt{{cn\ell}}$.	
\end{lemma}	

 \begin{proof}	
We prove the above lemma via an induction on the iteration number $i$. For the base case, i.e. at the start of the iteration $0$, there exist only singleton fragments which are of size at least $2^0$. For the induction hypothesis, we assume that the statement is true for iteration $i$, i.e. at the start of the $i^{th}$ iteration, the size of each fragment is at least $2^{i}$ and show that the statement also holds for iteration $i+1$, i.e. at the start of iteration $i+1$, the size of each fragment is at least $2^{i+1}$. To show this, consider all the fragments in iteration $i+1$, each fragment would either have diameter $\geq 2^{i+1}\ell$ or less than that. It is easy to see that if a fragment has diameter of $\geq 2^{i+1}\ell$, it would have greater than $\geq 2^{i+1}$ nodes, as latency of each edge is $\ell$. For the second case, where the fragment diameter is $< 2^{i+1}\ell$, we know from the algorithm (see line $3$ and $5$), that all such fragments merge with at least one more fragment. This other fragment is of size at least $2^i$ (from the induction hypothesis), and therefore the size of the resulting fragment at least doubles  i.e., becomes at least $2^{i+1}$. %

 Since fragments are disjoint, this implies that the number of fragments	
at the start of the $i^{th}$ iteration, is at most $n/2^i$. Thus, after $\log \sqrt{{n}/{c\ell}}$ iterations, the number of fragments is at most $\sqrt{{cn\ell}}$.	
\end{proof}	
}

\onlyLong{
\begin{lemma} \label{lem:base}	
Uniform Local Aggregation algorithm outputs at most $\sqrt{{cn\ell}}$ MST fragments each of diameter at most $O(\sqrt{n\ell/c})$ in $O(\sqrt{{n\ell}/{c}} \log^* n)$ rounds and requiring $O(m+ n \log^* n \log ({n}/{c\ell}))$ messages. 
\end{lemma}	
	
\begin{proof}	

 Each iteration of the uniform local aggregation algorithm performs three major functions, namely finding the MOE, convergecast within the fragment and merging with adjacent fragment over the matched MOE edge. 	
For finding the MOE, in each iteration, every node checks each of its neighbor (in $O(\ell)$ time) in non-decreasing order of weight of the connecting edge starting from the last checked edge (from the previous iteration). Thus, each node contacts each of its neighbors at most once, except for the last checked node (which takes one message per iteration). Hence total message complexity (over $\log \sqrt{{n}/{c\ell}}$ iterations) is 	
\[ \sum_{v\in V}2d(v) + \sum_{i=1}^{\log \sqrt{{n}/{c\ell}}} \sum_{v\in V}1 = m+n\left(\frac{1}{2} \log \frac{n}{c\ell}\right)\] \[ = O\left(m+ n \log \frac{n}{c\ell}\right) \] where $d(v)$ refers to the degree of a node. 	
The fragment leader determines the MOE for a particular iteration $i$, by convergecasting over the fragment, which requires at most $O(2^i\ell)$ rounds since the diameter of any fragment is bounded by $O(2^i\ell)$ (by Lemma \ref{lem:fragdiam}). 	
The fragment graph, being a rooted tree, uses a $O(\log^* n)$ round deterministic symmetry-breaking algorithm \cite{matching, Pandurangan:2017:TMD:3055399.3055449} to obtain the required matching edges in the case without latencies.	
Taking into account the required scale-up in case of the presence of latencies (one round for the non-latency case would be simulated as $\ell$ round in this case), the symmetry breaking algorithm	
is simulated by the leaders of neighboring fragments by communicating with	
each other; since the diameter of each fragment is bounded by $O(2^i\ell)$, the time	
needed to simulate one round of the symmetry breaking algorithm in the $i^{th}$ iteration	
is $O(2^i\ell)$ rounds. Also, as only the MST edges (MOE edges) are used in communication, the total	
number of messages needed is $O(n)$ per round of simulation. Since there are $\log \sqrt{{n}/{c\ell}} =	
O(\log ({n}/{c\ell}))$ iterations, the total time and message complexity for getting the	
maximal matching is $O(\sum^{ \log \sqrt{{n}/{c\ell}}}_{i=0} 2^i\ell \log^* n) = O(\sqrt{n\ell/c} \cdot \log^* n)$ and $O(n \log^* n)$	
respectively. Afterwards, adding selected edges into $M_i'$ (Line $5$ of the uniform local aggregation algorithm) can be done with additional $O(n)$ message complexity and $O(2^i\ell)$ time complexity in iteration $i$. The mergings are done over the matching edges and require one more round of convergecast to inform the nodes regarding the new fragment leader. This also takes $O(2^i\ell)$ time and $O(n)$ messages.	
Since there are $\log \sqrt{{n}/{c\ell}} = O(\log ({n}/{c\ell}))$ iterations in the 	
uniform local aggregation algorithm, the overall message complexity of the algorithm is	
$O(m+ n \log^* n \log ({n}/{c\ell}))$ and the overall time complexity is $O(\sqrt{{n\ell/c}} \cdot \log^* n)$.	
\end{proof}	
}
\onlyLong{	
\para{Global Aggregation Stage} \label{subsec:compmerge}	
Here components are merged using a BFS tree. %
The merging follow a similar procedure as in Section \ref{subsubsec:ubcorrelated}, however, the upcasting can now be done in a synchronous fashion as the latencies are uniform. We show that upcasting the MOE edges require  $O(D+\frac{1}{c} \sqrt{cn\ell})$ time and at most $O(D_{BFS}\cdot \sqrt{cn\ell})$ messages, where $D_{BFS}$ is the hop diameter of the BFS tree. However as $D$  is $\leq \sqrt{n\ell/c}$, it implies that $D_{BFS} \leq \sqrt{n/c\ell}$, which further implies that $O(D_{BFS}\cdot \sqrt{cn\ell}) = O(n)$. This trick of differentiating based on $D$ helps in limiting the messages to $\tilde{O}(m)$.}

\onlyLong{		
\begin{lemma} \label{lem:cm}	
For the case when $D \leq \sqrt{n\ell/c}$, merging components using the BFS tree requires $O((D+\sqrt{n\ell/c})\log n)$ rounds and ${O}(m\log n)$ messages.	
\end{lemma}

\begin{proof}	

 In the first step of the iteration, determining the MOE requires $O(\ell)$ time and $O(m)$ messages. Upcasting the MOE's to the base fragment leader requires $O(\sqrt{n\ell/c})$ rounds (base fragment diameter) and $O(n)$ messages (as only lightest MOE edge is upcast along a fragment). Next the base fragment leaders upcast to the root of the BFS tree in a synchronous fashion. For each base fragment ($\sqrt{cn\ell}-1$ many) only one MOE edge is upcast to the root of the shortest path tree. Therefore, the maximum possible time required for upcasting to the root is  $O(D+\frac{1}{c} \sqrt{cn\ell})$ and the number of messages sent is at most $O(D_{BFS}\cdot \sqrt{cn\ell})$, where $D_{BFS}$ is the hop diameter of the BFS tree. However as $D$ in this case is $\leq \sqrt{n\ell/c}$, it implies that $D_{BFS} \leq \sqrt{n/c\ell}$, which further implies that $O(D_{BFS}\cdot \sqrt{cn\ell}) = O(n)$. Once the root of the BFS tree obtains the MOE edges of all the components, it locally computes the component mergings (by locally simulating the uniform local aggregation algorithm) and thereafter informs all the fragment leaders of their updated component ids requiring $O(D+\frac{1}{c} \sqrt{cn\ell})$  and $O(D_{BFS}\cdot \sqrt{cn\ell}) = O(n)$ messages. Finally, the base fragment leaders inform all the nodes of the base fragment which requires at most $O(\sqrt{n\ell/c})$ rounds and $O(n)$ messages. %

 The cost of one iteration of merging using the BFS Tree is  $(O(\ell)+ O(\sqrt{n\ell/c}) + O(D+\sqrt{n\ell/c}))) = O(D+\sqrt{n\ell/c})$ time and $O(m+n+n) = O(m)$ messages. Since there are $\log n$ iterations the time complexity becomes $O((D+\sqrt{n\ell/c})\log n)$ and the message complexity is $O(m\log n)$.	
\end{proof}	
}

\onlyLong{
 The overall time and message complexity is determined by the cost of creating the base fragments along with the cost of merging the mst-components using the BFS tree. Therefore, the running time for the case $D \leq \sqrt{n\ell/c}$ is calculated as $O(\sqrt{{n\ell}/{c}} \log^* n) + O((D+\sqrt{n\ell/c})\log n) = {O}(\sqrt{n\ell/c} \log n)$ and the message complexity is $O(m+ n \log^* n \log ({n}/{c\ell}))+ O(m\log n) = {O}(m\log n+ n \log^* n \log ({n}/{c\ell}))$.	
}

\onlyLong{
 The case where $D> \sqrt{n\ell/c}$, the algorithm runs in a similar fashion except here the base fragments grow until their diameter equals the graph diameter $D$. In this case, the uniform local aggregation algorithm runs for $\log (D/\ell) \leq O(\log n)$  iterations (as all edge latencies are $\ell$, in worst case $D=n\ell$ ) and outputs at most $O(n\ell/D)$ fragments, requiring a running time of $O(D\cdot \log^*n)$ time and $O(m + n\log (D/\ell)\log^*n)$ messages. Thereafter, the mergings over the BFS tree require $O((D+ (1/c)(n\ell/D))\log n) = O((D+\sqrt{n\ell/c})\log n)$ time (as $D> \sqrt{n\ell/c}$) and $O((m+n+n)\log n)$ messages. We see that for this case as well, we obtain a time complexity of ${O}((D+ \sqrt{n\ell/c})\log n)$ and a message complexity of ${O}(m\log n + m + n\log^*n\log n)$. These results are proved using the following lemmas.
}

 \onlyLong{	
\begin{lemma} \label{lem:base2}	
Uniform Local Aggregation algorithm for the case where  $D> \sqrt{n\ell/c}$ runs for $\log (D/\ell)$  iterations and outputs at most $O(n\ell/D)$ MST fragments, each of diameter at most $O(D)$ in $O(D\cdot \log^*n)$ rounds and requiring $O(m + n\log^*n\log n)$ messages.	
\end{lemma}	

 \begin{proof}	

 From Lemma \ref{lem:fragdiam}, we see that for the uniform local aggregation algorithm at the start of the $i^{th}$ iteration, each fragment has diameter at most $O(2^i \ell)$.	
Therefore after  $\log (D/\ell)$ iterations, the fragment diameter would be at most $2^{\log D/\ell}\cdot\ell = O(D)$.
 From Lemma \ref{lem:fragsize}, we see that at the start of the $i^{th}$ iteration, each fragment has size at least $2^i$ and the number of fragments remaining is at most $n/2^i$. Therefore after  $\log (D/\ell)$ iterations here, the fragment size is at least $2^{\log D/\ell} = O(D/\ell)$ and the number of fragments remaining is at most $n/2^{\log D/\ell} = O(n\ell/D)$. Since here $D> \sqrt{n\ell/c}$, it implies that the number of fragments is $O(n\ell/D) \leq O(n\ell/\sqrt{n\ell/c}) = O(\sqrt{cn\ell})$.	

 Thereafter to determine the time and message complexity, we give a similar analysis as Lemma \ref{lem:base}	
Each iteration of the uniform local aggregation algorithm performs three major functions, namely finding the MOE, convergecast within the fragment and merging with adjacent fragment over the matched MOE edge. 	
For finding the MOE, in each iteration, every node checks each of its neighbor (in $O(\ell)$ time) in non-decreasing order of weight of the connecting edge starting from the last checked edge (from the previous iteration). Thus, each node contacts each of its neighbors at most once, except for the last checked node (which takes one message per iteration). Hence total message complexity (over $\log D/\ell$ iterations) is 	
\[ \sum_{v\in V}2d(v) + \sum_{i=1}^{\log D/\ell} \sum_{v\in V}1 = m+n\left( \log \frac{D}{\ell}\right)\] \[ = O\left(m+ n \log \frac{D}{\ell}\right) \] where $d(v)$ refers to the degree of a node. 	
The fragment leader determines the MOE for a particular iteration $i$, by convergecasting over the fragment, which requires at most $O(2^i\ell)$ rounds since the diameter of any fragment is bounded by $O(2^i\ell)$ (by Lemma \ref{lem:fragdiam}). 	
The fragment graph, being a rooted tree, uses a $O(\log^* n)$ round deterministic symmetry-breaking algorithm \cite{matching, Pandurangan:2017:TMD:3055399.3055449} to obtain the required matching edges in the case without latencies.	
Taking into account the required scale-up in case of the presence of latencies (one round for the non-latency case would be simulated as $\ell$ round in this case), the symmetry breaking algorithm	
is simulated by the leaders of neighboring fragments by communicating with	
each other; since the diameter of each fragment is bounded by $O(2^i\ell)$, the time	
needed to simulate one round of the symmetry breaking algorithm in iteration $i$	
is $O(2^i\ell)$ rounds. Also, as only the MST edges (MOE edges) are used in communication, the total	
number of messages needed is $O(n)$ per round of simulation. Since there are $O(\log D/\ell)$ iterations, the total time and message complexity for getting the maximal matching is $O(\sum^{ \log D/\ell}_{i=0} 2^i\ell \log^* n) = O(D \cdot \log^* n)$ and $O(n \log^* n)$	
respectively. Afterwards, adding selected edges into $M_i'$ (Line $5$ of the uniform local aggregation algorithm) can be done with additional $O(n)$ message complexity and $O(2^i\ell)$ time complexity in iteration $i$. The mergings are done over the matching edges and require one more round of convergecast to inform the nodes regarding the new fragment leader. This also takes $O(2^i\ell)$ time and $O(n)$ messages.	
Since there are $\log D/\ell$ iterations in the 	
uniform local aggregation algorithm when $D> \sqrt{n\ell/c}$, the overall message complexity of the algorithm is	
$O(m+ n \log^* n \log (D/\ell))$ and the overall time complexity is $O(D \log^* n)$.	
\end{proof}	
}	

 \onlyLong{	
\begin{lemma} \label{lem:cm2}	
For the case when $D > \sqrt{n\ell/c}$, merging components using the BFS tree requires $O((D+\sqrt{n\ell/c})\log n)$ rounds and ${O}(m\log n)$ messages.	
\end{lemma}

 \begin{proof}	

 In the first step of the iteration, determining the MOE requires $O(\ell)$ time and $O(m)$ messages. Upcasting the MOE's to the base fragment leader requires $O(D)$ rounds (base fragment diameter) and $O(n)$ messages (as only lightest MOE edge is upcast along a fragment). Next the base fragment leaders upcast to the root of the BFS tree in a synchronous fashion. For each base fragment only one MOE edge is upcast to the root of the shortest path tree. The total number of base fragments is $O(n\ell/D) \leq O(\sqrt{cn\ell})$. Therefore, the maximum possible time required for upcasting to the root is  $O(D+\frac{1}{c}\sqrt{cn\ell})$ and the number of messages sent is at most $O(D_{BFS}\cdot n\ell/D)$, where $D_{BFS}$ is the hop diameter of the BFS tree. As all edges have same latency, $D_{BFS}= D/\ell$. Therefore the total messages here is, $O(D/\ell\cdot n\ell/D) = O(n)$.	
Once the root of the BFS tree obtains the MOE edges of all the components, it locally computes the component mergings (by locally simulating the uniform local aggregation algorithm) and thereafter informs all the fragment leaders of their updated component ids requiring $O(D+\frac{1}{c} \sqrt{cn\ell})$  and $O(D_{BFS}\cdot n\ell/D) = O(n)$ messages. Finally, the base fragment leaders inform all the nodes of the base fragment which requires at most $O(D)$ rounds and $O(n)$ messages. %

 The cost of one iteration of merging using the BFS Tree is  $(O(\ell)+ O(D) + O(D+\sqrt{n\ell/c}))) = O(D+\sqrt{n\ell/c})$ time and $O(m+n+n) = O(m)$ messages. Since there are $\log n$ iterations the time complexity becomes $O((D+\sqrt{n\ell/c})\log n)$ and the message complexity is $O(m\log n)$.	
\end{proof}	
}

 \begin{theorem}	
In the weighted CONGEST model, when edges have uniform latencies, there exists a deterministic algorithm that computes the MST in ${O}((D + \sqrt{n\ell/c})\log n)$ rounds and with ${O}(m\log n + n \log^* n (\log ({n}/{c\ell})+ \log (D/{\ell}))$ messages.	
\end{theorem}	

 \onlyLong{	
\begin{proof}	
The correctness of the algorithm immediately follows from the fact that in each iteration MST fragments merge with one another, and the total number of fragments reduce by at least half. This ensures after the said many iterations, there is only one fragment remaining, which is in fact the MST.	
We calculate the time and message complexity of the algorithm by considering two cases. For either case, the overall cost is determined by the cost of creating the base fragments and thereafter the cost of merging mst-components using the BFS tree. As shown in Section \ref{spt}, we can create a shortest path tree (or a BFS tree in this case) in %
 $O(D)$ time and $O(m\log n)$ messages. 
 	
\noindent\textbf{Case 1: When $D \leq \sqrt{n\ell/c}$}  \\	
The cost of creating the base fragments is $O(\sqrt{\tfrac{n\ell}{c}} \log^* n)$ time and $O(m+ n \log^* n \log ({cn}/{\ell}))$ messages (c.f. Lemma \ref{lem:base}). Thereafter, merging the $\sqrt{{cn\ell}}$  mst-components using the BFS tree requires $O((D+\sqrt{n\ell/c})\log n)$ rounds and ${O}(m\log n)$ messages (c.f. Lemma \ref{lem:cm}).

 Therefore, the running time for the case $D \leq \sqrt{n\ell/c}$ is calculated as $O(D) + O(\sqrt{{n\ell}/{c}} \log^* n) + O((D+\sqrt{n\ell/c})\log n) = {O}((D + \sqrt{n\ell/c})\log n)$ and the message complexity is $O(m\log n) + O(m+ n \log^* n \log ({cn}/{\ell})) + O(m\log n) = {O}(m\log n + n \log^* n \log ({cn}/{\ell}))$.

\noindent\textbf{Case 2: When $D > \sqrt{n\ell/c}$} \\ 	
The cost of creating the base fragments is $O(D \log^* n)$ time and $O(m+ n \log (D/{\ell} \log^* n))$ messages (c.f. Lemma \ref{lem:base2}). Thereafter, merging the $O(n\ell/D)$  mst-components using the BFS tree requires $O((D+\sqrt{n\ell/c})\log n)$ rounds and ${O}(m\log n)$ messages (c.f. Lemma  \ref{lem:cm2}).

 Therefore, the running time for the case $D > \sqrt{n\ell/c}$ is calculated as $O(D) + O(D \log^* n) + O((D+\sqrt{n\ell/c})\log n) = {O}((D + \sqrt{n\ell/c})\log n)$ and the message complexity is $O(m\log n) + O(m+ n \log^* n \log (D/{\ell})) + O(m\log n) = {O}(m\log n + n \log^* n \log (D/{\ell}))$.	
\end{proof}	

 }
 
\subsection{Lower Bound}
We show that our proposed algorithm is near-optimal, by providing an almost tight lower bound up to polylogarithmic factors. 
We show the lower bound for the case where all edges have the same latency $\ell$ by giving a simulation that relates the running time of an algorithm with uniform latencies to that of the standard $\mathsf{CONGEST}$ model and subsequently leveraging on the Das Sarma et al. \cite{DHK+12} lower bound for the standard $\mathsf{CONGEST}$ model.
\onlyShort{The proof of the following theorem has been deferred to the full version \cite{augustine2019latency}.}

\begin{theorem} \label{thm:lb-l=w1}
Any algorithm to compute the MST of a network graph in which all edges have the same latency $\ell$ and capacity $c$ must, in the worst case, take $\Omega(D + ({\sqrt{n\ell/c}})/{\log n})$ time. 
\end{theorem}
\onlyLong{ 	
\begin{proof}

A lower bound of $\Omega(D)$ is trivial and follows immediately from the standard $\mathsf{CONGEST}$ model, albeit a scaling factor of $\ell$ due to the uniform latency edges. The scaling factor of $\ell$ appears while considering the hop diameter of the graph, however is absorbed back in the notation while considering latency diameter.  %

As all edges are of latency $\ell$ and capacity $c$, at any given instant of time, there can be at most $O(c\ell)$ messages on a link (i.e. at most $O(c\ell \log n)$ bits). This model is exactly equal to the standard $\mathsf{CONGEST}$ model where time is scaled by a factor of $\ell$ and with messages of size $O(c\ell \log n)$ bits.

From Das Sarma et al.~\cite{DHK+12}, we know that computing the MST in the CONGEST model with bandwidth $O(c\ell \log n)$ bits requires $\Omega(\sqrt{{n}/{B \log n}})$ rounds; where $B$ is the bandwidth term referring to the number of bits that can be sent over an edge per round. This, in turn, translates to a lower bound of $\Omega(\ell \cdot \sqrt{{n/{c \ell \log^2 n}}}) = \Omega(({\sqrt{(n\ell/c)}})/{\log n})$. 
\end{proof}
}

\section{Discussion} \label{sec:disc}
We have studied the problem of constructing an MST when the edges have a latency and a capacity associated with them. We provide several tight bounds on the time complexity depending on the relationship between the weights and the latencies. For the case where an edge weight equals its latency, we see that the algorithm is not always message optimal. For the case with unit latencies, several interesting recent works \cite{Pandurangan:2017:TMD:3055399.3055449, Elkin:2017:SDD:3087801.3087823, Haeupler:2018:RMD:3212734.3212737, ghaffari_et_al:LIPIcs:2018:9819} have shown that algorithms (both randomized and deterministic) can achieve simultaneous time and message optimality. These algorithms exploit the fact that in graphs with unit latencies both the latency diameter and the hop diameter are exactly the same. However, this does not hold for the case with arbitrary latencies, as there no longer exists a one-to-one correspondence between the latency diameter  and the hop diameter; and as such getting simultaneously optimal time and message complexities seems difficult. %
Simultaneous time and message optimality has been achieved in the classical \textsf{CONGEST} model with unit latencies. Whether such simultaneously optimal algorithms exist for our weighted \textsf{CONGEST} model remains an intriguing open problem.

To emphasize the challenge of optimizing both time and messages with latencies, consider the case of $k$-token aggregation on a given graph $G$, where initially $k$ different nodes possess a token and the goal is to collect these tokens at a predetermined sink node $s$. It is easy to see that the time complexity of doing so is $\Omega(D+k)$ and the corresponding message complexity is $\Omega(m+kD_{hop})$, where $D_{hop}$ is the hop diameter of $G$. For the case with unit latencies, both $D$ and $D_{hop}$ correspond to the exact same value and given any graph, a convergecast on a BFS tree with the sink node as the root gives an algorithm that is both time and message optimal. However, with latencies, there is an inherent trade-off.  Consider the case where the $k$ nodes having the token have a direct link to the sink node, albeit a link with a very high edge latency. Also, say there exists a long low latency path (consisting of multiple edges) from those $k$ nodes to the sink, that determines the latency diameter. Clearly, any time optimal algorithm taking $O(D+k)$ time would use the low latency path incurring a message cost of $O(m+k\cdot D)$; alternatively, any message optimal algorithm taking $(m+k)$ messages (as $D_{hop}=1$ due to the direct links) would need to use the high latency direct links which incur a cost of the latency time. Therefore, it is clear that for any $k$-aggregation protocol, in the presence of latencies, simultaneous time and message optimality is not possible. Thus one resolution to the open problem for MST construction would be to show that such an inherent trade-off exists there as well. (See \cite{gmyr_et_al:LIPIcs:2018:9821} for some interesting time-message trade-offs in the unit latency case.)

Another interesting open question regards capacities.  Here, we have considered the case where the capacity of every edge is exactly the same. However, the case where edges can have different edge capacities raises several questions.  For example, if the weights are equal to the capacities, the resulting MST would have useful data aggregation properties.

Several fundamental graph problems (e.g., those that rely on locality like MIS, matching and coloring) become interestingly more challenging when latencies are considered, leading to several open problems.\\

\bibliographystyle{abbrv}
\bibliography{bibliography}

\end{document}